\documentclass[nologo,notitlepageheader]{dakota-article}

\usepackage{caption}
\usepackage{subcaption}
\usepackage{natbib}

\usepackage{algorithm}
\usepackage{algorithmic}
\usepackage{hyperref}

\usepackage[textsize=tiny]{todonotes}

\usepackage{mathrsfs}

\newcommand{\reg}[1]{}

\newcommand{\als}[1]{}
\newcommand{\sgd}[1]{}
\newcommand{\opt}[1]{}

\makeatletter
\def\x#1{\def\tempa{#1}\futurelet\next\x@i}
\def\x@i{\ifx\next\bgroup\expandafter\x@ii\else\expandafter\x@end\fi}
\def\x@ii#1{x_{\tempa}^{(#1)}}
\def\x@end{x_{\tempa}}
\makeatother
\newcommand{\y}[1]{y^{(#1)}}

\usepackage{amsmath}
\DeclareMathOperator*{\argmin}{arg\,min}
\usepackage{amssymb}
\usepackage{bm}

\newcommand{\reals}{\mathbb{R}}
\newcommand{\integers}{\mathbb{Z}}
\newcommand{\xspace}[1]{ {\cal X}_{#1}}
\newcommand{\yspace}{ {\cal Y}}
\newcommand{\xdim}{d}
\newcommand{\ndata}{n}

\newcommand{\tens}[1]{\bm{\mathcal{#1}}}

\newcommand{\mat}[1]{\mathbf{#1}}
\newcommand{\bvec}[1]{\mathbf{#1}}

\newcommand{\pdimb}{p}
\newcommand{\pdim}{\pdimb_{t}}
\newcommand{\pdimp}[3]{\pdimb_{#1#2#3}}

\newcommand{\parambase}{\theta}
\newcommand{\param}{\bm{\parambase}}
\newcommand{\parami}[1]{\parambase_{#1}}
\newcommand{\paramuni}[4]{\parambase_{#1#2#3#4}}

\newcommand{\mvf}[1]{{\cal #1}}
\newcommand{\ftcore}[1]{{\mvf{F}_{#1}}}

\newcommand{\uni}[4]{{#1}_{#2}^{(#3#4)}}
\newcommand{\unif}[3]{\uni{f}{#1}{#2}{#3}}

\newcommand{\fspace}{\mathscr{F}}
\newcommand{\ftspace}{\fspace_{\bvec{r}}}

\newcommand{\multi}{\bm{\alpha}}

\newcommand{\ftexp}{\sum_{i_{0}=1}^{r_0}\sum_{i_{1}=1}^{r_1}\cdots \sum_{i_{d}=1}^{r_d} \unif{1}{i_0}{i_1}(x_1)\unif{2}{i_1}{i_2}(x_2)\ldots\unif{d}{i_{d-1}}{i_d}(x_d)}

\newtheorem{proposition}{Proposition}
\newtheorem{assumption}{Assumption}

\def\SNL{Optimization and Uncertainty Quantification, Sandia National Laboratories, Albuquerque, NM, 87123}
\def\UM{Department of Aerospace Engineering  University of Michigan  Ann Arbor, MI, 48109}
\def\WORK{This work was performed while the author was at Sandia National Laboratories}
\corauthor{A.A. Gorodetsky}
\coremail{goroda@umich.edu}
\author{A.A Gorodetsky\thanks{\UM} \thanks{\WORK}, J.D Jakeman\thanks{\SNL}}
\shortauthor{A.A. Gorodetksy, J.D Jakeman}
\shorttitle{Gradient-based optimization for low-rank functions}
\funding{See acknowledgements}
\title{Gradient-based optimization for low-rank functions}
\date{\today}
\title{Gradient-based Optimization for Regression in the Functional Tensor-Train Format}
\hypersetup{
  pdfkeywords={},
  pdfsubject={},
  pdfcreator={Emacs 25.2.1 (Org mode 8.2.10)}}
\begin{document}

\maketitle

\begin{abstract}
We consider the task of low-multilinear-rank functional regression, i.e., learning a low-rank parametric representation of functions from scattered real-valued data. Our first contribution is the development and analysis of an efficient gradient computation that enables gradient-based optimization procedures, including stochastic gradient descent and quasi-Newton methods, for learning the parameters of a functional tensor-train (FT). 
The functional tensor-train uses the tensor-train (TT) representation of low-rank arrays as an ansatz for a class of low-multilinear-rank functions. The FT is represented by a set of matrix-valued functions that contain a set of univariate functions, and the regression task is to learn the parameters of these univariate functions. Our second contribution demonstrates that using nonlinearly parameterized univariate functions, e.g., symmetric kernels with moving centers, within each core can outperform the standard approach of using a linear expansion of basis functions. Our final contributions are new rank adaptation and group-sparsity regularization procedures to minimize overfitting. We use several benchmark problems to demonstrate at least an order of magnitude lower accuracy with gradient-based optimization methods than standard alternating least squares procedures in the low-sample number regime. We also demonstrate an order of magnitude reduction in accuracy on a test problem resulting from using nonlinear parameterizations over linear parameterizations.  Finally we compare regression performance with 22 other nonparametric and parametric regression methods on 10 real-world data sets. We achieve top-five accuracy for seven of the data sets and best accuracy for two of the data sets. These rankings are the best amongst parametric models and competetive with the best non-parametric methods.
\end{abstract}


\section{Introduction}

Assesment of uncertainty in a computational model is essential to increasing the credibility of simulation based knowledge discovery, prediction and design. 
Sources of model uncertainty must be identified and the effect of these uncertainties on the model output (prediction) quantified. The accuracy to which uncertainty can be quantified is limited by the computational resources available. Many applications require vast amounts of computational effort, thereby limiting the number model evaluations that can be used to interrogate the uncertainty in the system behavior. Consequently a significant portion of methods developed for uncertainty quantification (UQ) in recent years have focused on constructing surrogates of expensive simulation models using only a limited number of model evaluations. 

Within the computational science community, both parametric and non-parametric function approximation methods have been extensively used for Uncertainty Quantification (UQ). Non parametric Gaussian process models (GP) \citep{Rasmussen2006,Ohagan_K_JRSSB_1978}  and parametric Polynomial Chaos Expansions (PCE) \citep{Ghanem_Book_1991,Xiu2002} are arguably the two most popular methods used. Gaussian process regression can be interpreted as a Bayesian method for function approximation, providing a posterior probability distribution over functions. Maximum likelihood estimation and Markov Chain Monte Carlo sampling are the two most popular methods used to characterize the posterior distribution of the GP.
Polynomial chaos expansions represent a response surface as a linear combination of orthonormal multivariate polynomials. The choice of the orthonormal polynomials is related to  the distribution of the uncertain model variables. Various approaches have been adopted to compute the coefficients of the PCE basis. Approaches include, pseudo-spectral projection \citep{conrad2013,constantine2012}, sparse grid interpolation \citep{gana07,nobile08b}, and regression using \(\ell_1\)-minimization \citep{blatman2011,doostan2011,mathelin2012}. For a comparison between nonparametric GP methods and parametric PCE methods see \citep{Gorodetsky2016}, and for an attempt to combine the benefits of both approaches see \citep{Schobi2015}.

High-dimensional approximation problems, such as regression, pose challenges for both parametric and nonparametric representation formats. Parametric approaches, for example those using a PCE representation,  are limited by their expressivity; and increasing expressivity, for example by increasing the polynomial order, results in the curse of dimensionality for fixed polynomial order. Nonparametric methods, for example Gaussian process regression, have great expressive capabilities. However, they also encounter the curse of dimensionality since their excess risk grows exponentially with dimension \citep{Gyorfi2002}. 

To counteract these computational burdens for both types of methods, attention has focused on constraining the functional representation to limit the curse of dimensionality. One popular constraint is limiting the model to that of additive separable forms \citep{Hastie1990,Meier2009,Lie2008}
\begin{equation}\label{eq:additive}
  f(x) = f_1(x_1) + f_2(x_2) + \ldots f_d(x_d).
\end{equation}
One can also use second order interactions, e.g., \(f_{12}(x_1,x_2), f_{13}(x_1,x_3), \ldots\), to increase expressivity while maintaining tractability \citep{Kandasamy2016}. However, further increasing the number of interaction terms in the ANOVA model \citep{Fisher1925}
\begin{equation*}
  f(x) = \sum_{i} f_i(x_i) + \sum_{i,j}f_{ij}(x_i,x_j) +  \sum_{i,j,k} f_{ijk}(x_i,x_j,x_k) + \cdots
\end{equation*}
will still naturally encounter the curse of dimensionality unless adaptive methods that identify the order of interaction interactively are utilized \citep{gana07,Ma_Z_JCP_2010,Foo_K_JCP_2010,Jakeman_R_SGA_2013,Jakeman_ES_JCP_2015}.

In this paper, we propose algorithms to improve regression in a functional representation that takes advantage of \emph{low-rank} structure to mitigate the curse of dimensionality while maintaining high expressivity. Low-rank functional representations are parametric models that enable a wide variety of interactions between variables and can generate high order representations. More specifically, they are continuous analogues of tensor decompositions and exploit \emph{separability}, i.e., that a function can be represented by the sum of products of univariate functions. One example is the canonical polyadic (CP) \citep{Carroll1970} representation consisting of a sum of products of univariate functions \(f(x) = \sum_{i=1}^R f_1^{(i)}(x_1)\ldots f_d^{(i)}(x_d)\), and the number of free parameters of such a representation scales linearly with dimension. Instead of the CP format, we use a continuous analogue of the discrete tensor train (TT) decomposition \citep{Oseledets2011} called the functional tensor-train (FT) \citep{Oseledets2013,Gorodetsky2015} to allow for a greater number of interactions between variables.

Low-rank functional tensor decompositions have been used for regression previously. Existing approaches \citep{Doostan2013,Mathelin2014,Chevreuil2015,Rauhut2017} implicitely make two simplifying assumptions to facilitate the use of certain algorithms and data structures from the low-rank tensor decomposition literature. Specifically, they assume linear and identical basis expansions for each univariate function of a particular dimension.  These approaches convert the problem from one of determining a low-rank \emph{function} to one of representing the coefficients of a tensor-product basis as a low-rank \emph{tensor}. Following this conversion, many of these techniques  use alternating minimization to determine the coefficients of the FT. Alternating minimization, such as alternating least squares, transforms a nonlinear optimization problem for fitting the parameters of each univariate function to data into a linear problem by only considering a single dimension at a time. Existing approaches either use efficient linear algebra routines to solve the linear system at each iteration \citep{Doostan2013} or sparsity inducing methods such as the LASSO \citep{Mathelin2014}. Recently, iterative thresholding has also been used to find low rank coefficients; however, such an approach has been limited to problems with low-dimensionality \citep{Rauhut2017}.

In this paper we take a different approach: we use gradient-based optimization procedures such as quasi-Newton methods and stochastic gradient descent, and we do not restrict ourselves by the assumptions that lead to the consideration of a tensor of coefficients. Overall, our contributions include:
\begin{enumerate}
\item Derivation and computational complexity analysis of a gradient-based fitting procedure that yields more accurate approximations than alternating minimization
\item Usage of both linear and nonlinear parameterizations of univariate function in each mode to facilitate a larger class of functions than is possible when using solely linear representations, and
\item Creation of rank-adaptation and regularization schemes to reduce overfitting.
\end{enumerate}

We demonstrate the efficacy of our approach on both synthetic functions and on several real-world data sets used in existing literature to test other regression algorithms \citep{Kandasamy2016}. Specifically, we show that gradient-based optimization provides significant advantages in terms of approximation error, as compared with an alternating least squares approach, and that these advantages are especially apparent in the case of small amounts of data and large parameter sizes. To our knowledge, no gradient based procedure has been derived or evaluated for the low-rank functional format that we consider. We also show the benefits of nonlinear representations of univariate functions over the linear representations most frequently found in the literature. Our real-world data results show that our methodology is competitive with both nonparametric and parametric algorithms. We achieve top-five accuracy for seven of the 10 data sets considered and best accuracy for two of the data sets. These rankings are the best amongst parametric models and competitive with the best non-parametric methods. Finally we demonstrate that for some physical systems, exploiting low-rank structure is more effective than exploiting sparse structure which is common approach used for enabling uncertainty quantification of computational models.

\subsection{Related Work}

As mentioned above, the functional tensor-train decomposition was proposed as an ansatz for representing functions by \cite{Oseledets2013}, and computational routines for compressing a function into FT format have been developed before \citep{Gorodetsky2015}. In that setting, an approximation of a black-box function is sought to a prescribed accuracy. A sampling procedure and associated algorithm was designed to optimally evaluate the function in order to obtain an FT representation. In this work, we consider the setting of fixed data.

There has also been some recent work on regression in low-rank format \citep{Doostan2013,Mathelin2014,Chevreuil2015}. These approaches rely on linearity between the parameters of the low-rank format and the output of the function. Utilizing this relationship, they convert the low-rank function approximation to one of low-rank tensor decomposition for the coefficients of a tensor-product basis. In Section \ref{sec:lrf_vs_lrp} we show how the representation presented in those works can be obtained as a particular instantiation of the formulation we present here.

In spirit, our approach is also similar to the recent use of the tensor-train decomposition within fully connected neural networks \citep{Novikov2015}. There, the layers of a neural network are assumed to be a linear transformation of an input vector, and the weight matrix of the transformation is estimated from data. Their contribution is representing the weight matrix in a compressed TT-format. In this work, our low-rank format can be thought of as an alternative to the linear transformation layer of a neural network. Indeed, future work can focus on utilizing our proposed methodology within a hierarchy of low-rank approximations.

\section{Background}
In this section we establish notation and provide background for regression and low-rank tensor decompositions.

\subsection{Notation}
Let $\reals$ be the set of real numbers and \(\integers_{+}\) be the set of positive integers. Let \(\ndata \in \integers_{+}\), \(\xdim \in \integers_{+}\) and suppose that we are given i.i.d. data $\left(\x{}{i},\y{i}\right)_{i=1}^n$ such that each datum is sampled from a distribution $\mu_{xy}$ on a compact space \((\x{}{i},\y{i}) \in \xspace{} \times \yspace \subset \reals^{d} \times \reals.\) Let $\xspace{} = \xspace{1} \times \cdots \times \xspace{d}$  be a tensor product space with $\xspace{k} \subset \reals$. Let the marginal distribution of $\x{} = (\x{1},\ldots,\x{d}) \in \xspace{}$ be the tensor product measure $\mu_{x} = \mu_{\x{1}} \times \cdots \times \mu_{\x{d}}$, and assume that all integrals appearing in this paper are with respect to this measure. For example let $f:\xspace{} \to \yspace$ and $g:\xspace{} \to \yspace$, then the inner product is defined as \(\left\langle f,g \right\rangle = \int f(x)g(x)d\mu(x).\) Similarly, the $L_2$ norm is defined as \(\lVert f \rVert^2_2 = \left\langle f,f \right\rangle.\)

In this paper scalars and scalar-valued functions are denoted by lowercase letters, vectors are denoted by bold lower-case letters such as $\bvec{x},\bvec{y}$;  matrices are denoted with upper boldface such as $\mat{X},\mat{Y}$; tensors are denoted with upper boldface caligraphic letters such as $\tens{X},\tens{Y}$; and matrix-valued functions are denoted by upper caligraphic letters such as \(\mvf{F},\mvf{G}.\) An ordered sequence of elements of the same set are distinguished by parenthesized superscripts such as \(\x{}{i}\) or \(\y{i}\).
\subsection{Supervised learning}
The supervised learning task seeks a function \(f: \xspace{} \to \yspace\) that minimizes the average of a cost function \(g_f :\xspace{} \times \yspace{} \to \reals\)
\begin{equation} 
  J[f] = \int_{\xspace{} \times \yspace} g_f(x,y) d\mu(x,y),
\end{equation}
For example, the standard least squares objective is specified with \(g_f(x,y) = \left(y - f(x)\right)^2\). In this work, the cost functional cannot be exactly calculated; instead, data \(\left(\x{}{i},\y{i}\right)_{i=1}^n\) is used to estimate its value. The cost functional then becomes a sum over the data instead of an integral
\begin{equation} \label{eq:ls}
  J\left[f\right] = \frac{1}{\ndata} \sum_{i=1}^{\ndata} \left(\y{i} - f\left(\x{}{i}\right)\right)^2,
\end{equation}
where we have reused the notation \(J[f]\), and further references to \(J\) will use this definition.
To obtain an optimization problem over a finite set of variables, the search space of functions must be restricted. Nonparametric representations generally allow this space to vary depending upon the data, and parametric representations typically fix the representation. In either case, we denote this space as \(\fspace\) to seek an optimal function \(f^* \in \fspace\) such that

\begin{equation}\label{eq:fspace_learn}
  f^* = \arg \min_{f \in \fspace} J[f].
\end{equation}

One example of a common function space is the reproducing kernel Hilbert space, and resulting algorithms using this space include Gaussian process regression or radial basis function interpolation. Other examples include linear functions (resulting in linear regression) or polynomial functions. In this work, we consider a function space that incorporates all functions with a prescribed \emph{rank}, which will be defined in the next section.

When solving the supervised learning problem \eqref{eq:fspace_learn} it is often useful to introduce a regularization term to minimize overfitting. In this paper we will consider the following regularized learning problem
\begin{equation}\label{eq:fspace_learn_reg}
  f^* = \arg \min_{f \in \fspace} J[f]+\lambda \Omega[f],
\end{equation}
where \(\lambda \in \reals_{+}\) denotes a scaling factor and \(\Omega: \fspace \to \reals_{+}\) is a functional that penalizes certain functions in \(\fspace\). For example, the function \(\Omega(f) = \lVert f \rVert_2\) penalizes functions that have large \(L_2\) norms, and such a penalty has been used for a certain type of low-rank functional approximation before \citep{Doostan2013}. In this work, we  impose a group sparsity type regularization that has been shown to improve the prediction performance in other approximation settings \citep{Turlach_VW_Technometrics_2005,Yuan_L_JRSSSB_2006}. Such an approach seeks to increase parsimony by reducing the number of non-zero elements in an approximation. In Section \ref{sec:lowrank-constraint} we describe what this type of constrains means in the context of low-rank functional decompositions.

\subsection{Low-rank tensor decompositions}
The function space that we use to constrain our supervised learning problem is related to the concept of low-rank decompositions of multiway arrays, or \emph{tensors}. Tensor decompositions are multidimensional analogues of the matrix SVD and are used to mitigate the curse of dimensionality associated with storing multivariate arrays \citep{Kolda2009}.  Consider an array with \(\tens{A} \in \reals^{\pdimb_1 \times \cdots \times \pdimb_d}\), this array contains an exponentially growing number of elements as the dimension increases. Tensor decompositions can provide a compact representation of the tensor and have found wide spread usage for data compression and learning tasks \citep{Cichocki2009,Novikov2015,Yu2016}.

In this work, we consider the tensor-train (TT) decomposition \citep{Oseledets2011}. Specifically, we use the TT as an ansatz for representating low-rank functions. A TT representation of a tensor \(\tens{A}\) is defined by a list of \(3\)-way arrays \(\left(\tens{A}_k \in \reals^{r_{k-1} \times \pdimb_{k} \times r_{k}}\right)_{k=1}^{d}\) called TT-cores where \(r_0=r_d=1\) and \(r_k \in \integers_+\) for \(k=2,\ldots,d-1\). The sizes  of the cores \(\left(r_k\right)_{k=0}^d\) are called the \emph{TT-ranks}. In this format a tensor element is obtained according to the following multiplication
\begin{equation} \label{eq:tt}
\tens{A}[i_1,i_2,\ldots,i_d] = \tens{A}_{1} [i_1] \tens{A}_{2} [i_2] \cdots \tens{A}_d[i_d], \quad  1 < i_k < \pdimb_{k} \textrm{ for all } k,
\end{equation}
where \(\tens{A}_k[i_k] \in \reals^{r_{k-1} \times r_{k}}\) are matrices and the above equation describes a sequence of matrix multiplication.

In this format, storage complexity scales \textit{linearly} with dimension and quadratically with TT-rank. 
In the next section we discuss how it can be used as a framework for function approximation.

\section{Low-rank functions}

Low-rank formats for functions can be thought of as direct analogues to low-rank tensor decompositions. In this work, we focus on aa TT-like decomposition of functions called the functional tensor-train (FT) \citep{Oseledets2013}.
 \begin{equation}
  f(\x{1},\x{2},\ldots,\x{d}) = \ftexp,\label{eq:ftlong}
\end{equation}
where \(\uni{f}{k}{i}{j}:\xspace{k} \to \reals\), and \(r_0=r_d=1\) for single-output functions. A more compact expression is obtained by viewing a function value as a set of products of matrix-valued functions
\begin{equation}\label{eq:ft}
  f(\x{1},\x{2},\ldots,\x{d}) = \mvf{F}_1(x_1)\mvf{F}_2(x_2)\ldots \mvf{F}_d(x_d),
\end{equation}
where each matrix-valued function \(\mvf{F}_k:\xspace{k} \to \reals^{r_{k-1} \times r_k}\) is called a \textit{core} and can be visualized as an array of the univariate functions
\begin{equation}\label{eq:ftcore}
  \mvf{F}_{k}(x_k) = \left[
    \begin{array}{ccc}
      \unif{k}{1}{1}(\x{k}) & \cdot &\unif{k}{1}{r_{k}}(\x{k}) \\
      \vdots & \ddots & \vdots \\
      \unif{k}{r_{k-1}}{1}(\x{k}) & \cdot &\unif{k}{r_{k-1}}{r_{k}}(\x{k}) 
    \end{array}
  \right].
\end{equation}
If each univariate function is represented with $\pdimp{}{}{}$ parameters and \(r_k < r\) for all \(k\), then the storage complexity scales as \(\mathcal{O}(dpr^2)\). Comparing this representation with \eqref{eq:tt} we see a very close resemblence between the TT cores and the FT cores. Indeed they are both matrices when indexed by a discrete index \(i_k\) for the TT  or a continuous index \(x_k\) for the FT. We describe a closer relationship between the TT and the FT in the context of low-rank representations of functions in the next section.

\subsection{Parameterizations of low-rank functions}
\label{sec:params} 
An FT core is comprised of \(\xdim\) sets of univariate functions as shown in Figure \ref{fig:fthierarchy}. Each set of univariate functions, contains all of the parameters of the associated univariate functions. As a result the full FT is parameterized through the parameterization of its univariate functions. Let $\pdimp{k}{i}{j} \in \integers_+$ denote the number of parameters describing \(\unif{k}{i}{j}.\) Let $\param$ denote the vector of parameters of all the univariate functions. Then, there are a total of $\sum_{k=1}^{\xdim}\sum_{i=1}^{r_{k-1}}\sum_{j=1}^{r_k}\pdimp{k}{i}{j}$ parameters describing the FT, i.e., $\param \in \reals^{\pdim}$.

The parameter vector $\param$ is indexed by a multi-index \(\multi = (k,i,j,\ell)\) where \(k = 1,\ldots,\xdim\) corresponds to an input variable, \(i = 1,\ldots,r_{k-1}\) and \(j=1,\ldots,r_{k}\) correspond to a univariate function within the \(k\)th core, and \(\ell = 1,\ldots,\pdimp{k}{i}{j}\) corresponds to a specific parameter within that univariate function. In other words, we adopt the convention that $\param_{\multi} = \param_{kij\ell}$ refers to the \(\ell\)th parameter of the univariate function in the \(i\)th row and \(j\)th column of the \(k\)th core.

The additional flexibility of the FT representation allows both linear and nonlinear parameterizations of univariate functions. A linear parameterization represents a univariate function as an expansion of basis functions \(\left(\uni{\phi}{k\ell}{i}{j} : \xspace{k} \to \reals \right)_{\ell=1}^{\pdimp{k}{i}{j}}\) according to
\begin{equation}
  \unif{k}{i}{j}(\x{k};\param) = \sum_{\ell=1}^{\pdimp{k}{i}{j}}\paramuni{k}{i}{j}{\ell} \uni{\phi}{k\ell}{i}{j}(\x{k}).
\label{eq:linparam}
\end{equation}
Linear parameterizations are often convenient within the context of gradient-based optimization methods since the gradient with respect to a parameter is independent of the parameter values. Thus, one only needs to compute and store the evaluation of the basis functions a single time. Nonlinear parameterizations are more flexible and general, but often incur a greater computational cost within optimization. One example of a nonlinear parameterization is that of a Gaussian kernel, which can be written as 
\begin{equation}
  \unif{k}{i}{j}(\x{k};\param) = \sum_{\ell=1}^{\pdimp{k}{i}{j}/2}\paramuni{k}{i}{j}{\ell} \exp\left(-\frac{1}{\sigma^2}\left(\x{k} - \paramuni{k}{i}{j}{(\pdimp{k}{i}{j}/2+\ell)}\right)^2 \right), \label{eq:nonkernel}
\end{equation}
where \(\sigma > 0\)
Here, the first $\pdimp{k}{i}{j}/2$ parameters refer to the coefficients of radial basis functions and that second half of the parameters refer to the centers. 

\begin{figure}
  \includegraphics[width=\textwidth]{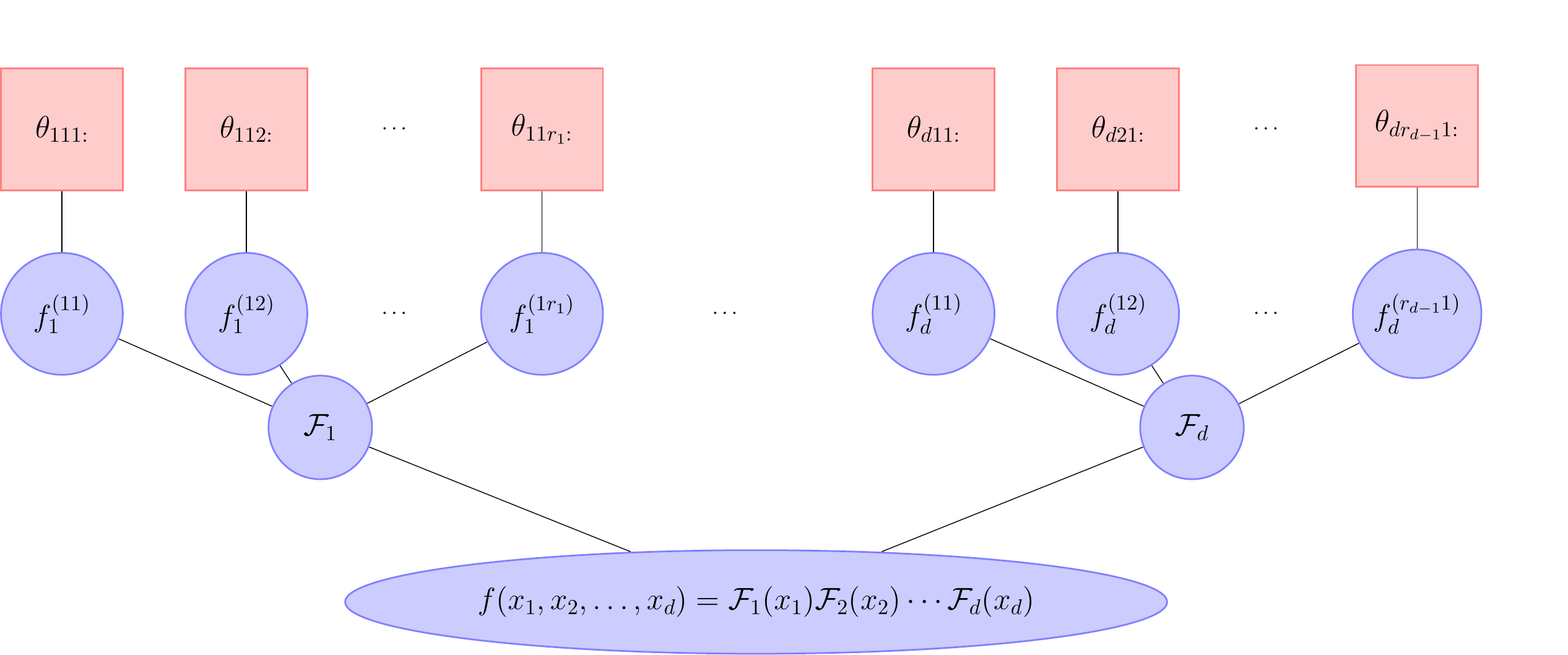}
  \caption{Object-oriented hierarchy for storing low-rank functions. Parameters are represented by red squares and functions are represented in blue.}
  \label{fig:fthierarchy}
\end{figure}

\subsection{Low-rank functions vs. low-rank coefficients}
\label{sec:lrf_vs_lrp}
The functional tensor train can be used by independently parameterizing the univariate functions of each core, and both linear and nonlinear parameterizations are possible. As described below, the advantage of this representation includes a naturally sparse storage scheme for the cores. Another advantage is the availablitiy of efficient computational algorithms for multilinear algebra that can adapt the representation of each univariate function individually as needed \citep{Gorodetsky2015,Gorodetsky2017} in the spirit of continuous computation pioneered by Chebfun \citep{Platte2010}.

Although the functional form of the TT is general, most of the literature makes two simplifying assumptions \citep{Doostan2013,Mathelin2014,Chevreuil2015}:
\begin{enumerate}
\item a \emph{linear} parameterization of each \(\unif{k\ell}{i}{j}\), and
\item an \emph{identitical basis} for the functions within each FT core, i.e., \(\pdimp{k}{i}{j} = \pdimp{k}{}{}\) and \(\uni{\phi}{k\ell}{i}{j} = \phi_{k\ell}\) for all \(i = 1, \ldots r_{k-1}\), \(j = 1, \ldots, r_{k},\) and \(\ell = 1,\ldots,\pdimp{k}{}{}\).
\end{enumerate}
These assumptions transform the problem of learning low-rank functions to the problem of learning low-rank coefficients. This transformation of the problem allows one to facilitate the use of discrete TT algorithms and theory. We will refer to representations using these assumptions as a functional tensor-train with TT coefficients or FT-c, where ``c'' stands for coefficients.

The FT-c representation stores the coefficients of a tensor-product basis \(\phi_{k\ell}\) for all \(k\) and \(\ell\) in TT format. Let \(\tens{F}_k \in \reals^{r_{k-1} \times \pdimp{k}{}{} \times r_{k}}\) be a tensor of the following form
\begin{equation}\label{eq:ttcore}
 \tens{F}_k[:,\ell,:] = \left[
    \begin{array}{ccc}
      \paramuni{k}{1}{1}{\ell} & \cdot & \paramuni{k}{1}{r_{k}}{\ell} \\
      \vdots & \ddots & \vdots \\
      \paramuni{k}{r_{k-1}}{1}{\ell} & \cdot & \paramuni{k}{r_{k-1}}{r_{k}}{\ell}
    \end{array}
  \right],
\end{equation}
for \(\ell = 1, \ldots, \pdimp{k}{}{}\). Function evaluations can be obtained from the from coefficients stored in TT format by performing the the following summation
\begin{equation} \label{eq:low_rank_coeffs}
  f(x_1,\ldots,x_{\xdim}) = \sum_{\ell_1=1}^{\pdimp{1}{}{}} \cdots \sum_{\ell_{\xdim}=1}^{\pdimp{\xdim}{}{}} \tens{F}_1[:,\ell_1,:] \cdots \tens{F}_{\xdim}[:,\ell_{\xdim},:] \phi_{1\ell_1}(x_1) \cdots \phi_{\xdim \ell_{\xdim}}(x_{\xdim}).
\end{equation}
From \eqref{eq:ft}, \eqref{eq:ttcore}, and \eqref{eq:low_rank_coeffs} we can see that the relationship between the TT cores \(\tens{F}_k\) and the FT cores \(\mvf{F}_k\) is
\begin{equation}\label{eq:tt_to_ft}
\mvf{F}_k(x_k) = \sum_{\ell=1}^{\pdimp{k}{}{}} \tens{F}_k[:,\ell,:]\phi_{k\ell}(x_k),
\end{equation}
where the basis function multiplies every element of the tensor. In other words, the TT cores represent a TT decomposition of the \(p_1 \times p_2 \times \cdots \times p_d\) coefficient tensor of a tensor-product basis.

The two assumptions required for converting the problem from one of low-rank function approximation to low-rank coefficient approximation often result in a computational burden in practice. The burden of the first assumption is clear, some functions are more compressible if nonlinear parameterizations are allowed. The second assumption can also limit compressibility, but also results in a larger computational and storage burden. One example in which this burden becomes obvious is when attempting to represent \emph{additively separable} functions, e.g., Equation \eqref{eq:additive}, in the FT format. As mentioned in the introduction, this representation is common for high-dimensional regression, and an FT can represent this function using cores that fave the following rank 2 structure.
\begin{align*}
 f(x_1,x_2,\ldots,x_{\xdim}) = \left[f_1(x_1) \ 1 \right] \left[ \begin{array}{cc} 1 & 0\\ f_2(x_2) & 1 \end{array} \right] \cdots \left[ \begin{array}{c} 1  \\ f_{\xdim}(x_{\xdim}) \end{array}\right].
\end{align*}
Suppose that each of the univariate functions can be represented with $\pdimp{k}{}{}$ parameters and the constants $0$ and $1$ can be stored with a single parameter. Then, the storage requirement is $\pdimp{k}{}{}+3$ parameters for the $d-2$ middle cores and $\pdimp{k}{}{} + 1$ for the outer cores, totaling $(\pdimp{k}{i}{j}+3)\xdim-4$ floating point values. In the TT case, the $0$ and $1$ terms must be stored with the \emph{same} complexity as the other terms, since the TT is a three-way array. Thus, the total number of parameters becomes $4\pdimp{k}{}{}(\xdim-2)$. Almost four times less storage is required for the FT format than the TT format, in the limit of large $\pdimp{k}{}{}$ and $d$. Essentially, the TT format does take into account the \emph{sparsity} of the cores. In this case, one can think of the FT format as storing the TT cores in a \emph{sparse} format. This burden is exacerbated if we add interaction terms to Equation \ref{eq:additive}.

\section{Low-rank supervised learning}
\label{sec:lowranksup}

In this section, we incorporate the FT model as a constraint in the supervised learning task. We discuss issues surrounding optimization algorithms, regularization, and choosing the TT-ranks. In particular we discuss three optimization algorithms for fitting fixed-rank models: batch gradient descent, alternating least squares, and stochastic gradient descent. We also present approaches for minimizing over-fitting. Specifically we discuss a group-sparsity-inducing regularization regularization scheme and hyper-parameter estimation scheme using cross validation and a rank-adapation. 

\subsection{Low-rank constraints and regularization}
\label{sec:lowrank-constraint}
The function space $\fspace$ in Equation \eqref{eq:fspace_learn} constrains the search space. The space constrained by low-rank functions $\ftspace$ is defined as follows.
Let \(\bvec{r} = [1,r_1,\ldots,r_{d-1},1]\) such that $r_i \in \integers_+$ for $i = 1,\ldots,d-1$. Then
\begin{equation*}
  \ftspace = \left\{
  f: \xspace{} \to \reals \left|
  \begin{array}{l}
    f = \mvf{F}_1(x_1)\mvf{F}_2(x_2)\ldots \mvf{F}_d(x_d) \\
    \mvf{F}_1:\xspace{1} \to \reals^{1 \times r_1} \\ 
    \mvf{F}_k:\xspace{k} \to \reals^{r_{k-1} \times r_{k}}, \  2 \leq k < d \\
    \mvf{F}_d:\xspace{d} \to \reals^{r_{d-1} \times 1} \\
  \end{array}
  \right.
  \right\}
\end{equation*}
denotes the space of functions with FT ranks \(\bvec{r}.\) Note that this function space also includes functions with smaller ranks. For example, a function that differs in ranks $\hat{r}_1,\hat{r}_2$ such that  $\hat{r}_1 < r_1$ and $\hat{r}_2 < r_2$ can be obtained by setting all univariate functions, aside from the top left $\hat{r}_1 \times \hat{r}_2$ block, to zero; analogous modifications must be made to all other cores.

In this paper we also add a regularization to attempt to minimize the number of nonzero functions in each core. The structure of the FT-cores in Equation \eqref{eq:ftcore} admits a natural grouping in terms of the univariate functions. Setting a univariate function to zero inside of a core essentially lowers the number of interactions allowable between  univariate functions in neighbouring dimensions, with the overall effect beign a smaller number of sums in the representation of the function in Equation \eqref{eq:ftlong}. Minimizes the number of nonzero univariate functions not only produces a more parsimonious model, but also minimizes the number of interactions between function variables. 

Specifically, we penalize the regression problem using the sum of the norms of the univariate functions 
 \begin{equation}\label{eq:tt_reg}
\Omega[f] = \sum_{k=1}^d\sum_{i=1}^{r_{k-1}}\sum_{j=1}^{r_{k}}\lVert \unif{k}{i}{j} \rVert^2.
 \end{equation}
Note that this regularization method is different from that first proposed in \citep{Mathelin2014} where the FT-c representation was used and the TT-cores of the coefficient tensor themselves were forced to be sparse. In our case we do not look for sparse coefficients, instead our goal is to increase the number of functions that are identically zero. As a surrogate for this goal we follow the standard practice of replacing the non-differential ``\(L_0\)" norm by the another norm, in our case a sum of the norms of the functions. This replacement also seeks to directly minimize the number of interacting univariate functions and also provides a differentiable objective for optimization.

\subsection{Hyper-parameter estimation using cross validation and rank adaptation}
\label{sec:rankadapt}

A careful selection of the number of parameters \(\pdimp{k}{i}{j}\) in each univariate function, the rank \(r_k\) of each core, and the Lagrangian multiplier \(\lambda\) in the regularized learning problem is required to avoid over fitting. For example if a polynomial basis is chosen, setting the degree too high can result in spurious oscillations.  Cross validation provides a mechanism to estimate hyper-parameters. In this paper we use 20 fold cross validation to select the hyper-parameters of an FT that minimizes the expected prediction error.

Alternatively, we can design a more effective rank adaptation scheme by combining cross validation with a rounding procedure. The problem with the simplest cross validation option is that a scheme that optimizes separately over each $r_k$ imposes a computational burden, and a scheme that optimizes for a single rank across all cores, i.e., $r_k=r$ does not allow enough flexibility. Instead, we propose to combine FT rounding \citep{Gorodetsky2015} and cross validation to avoid overfitting. 

Rounding is a procedure to reapproximate an FT by one with lower ranks to a given tolerance. The tolerance criterion can be thought of as a regularization term, since it limits the complexity of the representation and will be investigated in furture work. The full rank adaptation scheme is provided by Algorithm \ref{alg:rankadapt}. In that algorithm $\texttt{cv}(\bvec{r})$ refers to a function that provides a cross-validation error for an optimization over the space $\ftspace$, and the function $\texttt{rounding-rank}(f,\delta)$ provides the ranks of a rounded function $f$ with a particular tolerance $\delta$.

\begin{algorithm}
  \caption{Rank adaptation}
  \label{alg:rankadapt}
  \begin{algorithmic}[1]
    \REQUIRE Rounding tolerance $\delta$
    \STATE $\bvec{r} = \texttt{ones}(\xdim+1)$
    \STATE $\epsilon = \texttt{cv}(\bvec{r})$
    \WHILE{not converged}
      \FOR{$k = 1, \ldots d-1$}
        \STATE $r_k \leftarrow r_k + 1$ 
        \ENDFOR
      \STATE $\bvec{r} = [1, r_1, \ldots, r_{d-1}, 1]$
      \STATE $\hat{\epsilon} = \texttt{cv}(\bvec{r})$
      \IF {$\hat{\epsilon} > \epsilon$}
         \STATE $\bvec{r} = [1, r_1-1, \ldots, r_{d-1}-1, 1]$
         \STATE \textrm{break}
         \ENDIF
      \STATE $\epsilon = \hat{\epsilon}$
      \STATE $f = \min_{f \in \ftspace} J(\param)$
      \STATE $[1, \hat{r}_1, \ldots, \hat{r}_{d-1}, 1] = \texttt{rounding-rank}(f,\delta)$
      \IF {$\hat{r}_k < r_k$ for all $k=1,\ldots,d-1$}
         \STATE \textrm{break}
      \ELSE
         \STATE $r_k = \hat{r}_k$ for $k=1,\ldots,d-1$
      \ENDIF
%
    \ENDWHILE
  \end{algorithmic}
\end{algorithm}

The scheme increases ranks until either the cross-validation error increases \emph{or} until the rounding procedure decreases every rank. The first termination criterion targets overfitting, and the second termination criterion seeks to limit rank growth when data is no longer informative enough to necessitate the increase in expressivity.

\subsection{Optimization algorithms}

In this section, we overview three gradient-based optimization algorithms for solving supervised learning problems in low-rank format: alternating minimization, gradient decent and stochastic gradient decent. The computational complexity of these algorithms in analyzed in Section \ref{sec:analysis}.

\subsubsection{Alternating minimization}
\label{sec:als}

Alternating minimization methods, mainly alternating least squares, are the main avenue for optimization in low-rank tensor and functional contexts. These routines are typically used within tensor decompositions by sequentially sweeping over dimensions and optimizing over parameters of a single dimension. Such approaches are popular for compressing multiway arrays or for tensor completion problems \citep{Savostyanov2011,Grasedyck2015}.

For the case of supervised learning, their usage is straightforward. The idea is to solve a sequence of optimization problems, where the functional space \(\ftspace\) is further constrained by fixing \emph{all-but-one} FT core, \(\mvf{F}_k\). After optimizing over the \(k\)th core, that core is fixed and optimization over the next one is performed. This algorithm is provided by Algorithm \ref{alg:als}. This algorithm performs sweeps over all dimensions until convergence is obtained. For details on convergence of this algorithm we refer to, e.g., \cite{Uschmajew2012}. We will provide more details about the implementation and complexity of this algorithm in Section \ref{sec:analysis}.

\begin{algorithm}
  \caption{Alternating minimization for low-rank supervised learning}
  \label{alg:als}
  \begin{algorithmic}[1]
    \REQUIRE Parameterized low-rank function space, $\ftspace$; Initial FT cores $\mvf{F}_k$ for $k=1,\ldots,d$; Data, $\left(\x{}{i},\y{i}\right)_{i=1}^N$; Objective function, $J$; Convergence tolerence, $\epsilon$
    \STATE $i=0$
    \STATE $f^{(i)} = \mvf{F}_1 \cdots \mvf{F}_d$
    \WHILE{not converged}
    \FOR{$k = 1, \ldots d$}
    \STATE $\displaystyle{\mvf{F}_k = \argmin_{\hat{\mvf{F}}_k} J[\mvf{F}_1\cdots \mvf{F}_{k-1} \hat{\mvf{F}}_{k} \mvf{F}_{k+1}\cdots\mvf{F}_d]}$ \label{alg:als:line:argmin}
    \ENDFOR
    \STATE $f^{(i+1)} = \mvf{F}_1 \cdots \mvf{F}_d$
      \IF{ $\lVert f^{(i+1)} - f^{(i)} \rVert \leq \epsilon$}
      \STATE break;
      \ENDIF
    \STATE $i = i+1$
    \ENDWHILE
  \end{algorithmic}
\end{algorithm}

\subsubsection{Batch gradient methods}
\label{sec:allatonce}

Gradient descent directly minimizes the cost function $J(\param)$ with a batch gradient (or second-order) based procedure. While this is a standard optimization approach, we will refer to this algorithm as \emph{all-at-once} (AAO) to distinguish it from the alternating minimization. Gradient-based  procedures have been shown to be promising in the context of canonical tensor decompositions and gradient descent  \citep{Acar2011} and TT decompositions in the context of iterative thresholding \citep{Rauhut2017}, but have not been explored well for low-rank functional approximation.

Gradient descent updates parameters according to
\begin{equation*}
  \param \leftarrow \param - \eta \nabla J(\param),
\end{equation*}
for \(\eta > 0.\) More generally, we use a quasi-Newton method to choose a direction \(\bvec{p}\) such that the update becomes
\begin{equation*}
  \param \leftarrow \param - \eta \bvec{p}
\end{equation*}
In the examples in Section \ref{sec:experiments}, we use the limited-memory BFGS \citep{Liu1989} method available in the $C^3$ library\footnote{github.com/goroda/Compressed-Continuous-Computation} to perform this update.

One disadvantage of this approach that is often cited is that it involves solving an optimization problem whose size is the number of parameters in the function. However, we note that the number of parameters scales as \(\mathcal{O}(dr^2\pdimp{}{}{}),\) so for a one hundred dimensional, rank 5 problem with \(\pdimp{k}{i}{j}=\pdimp{}{}{}=5\) we have \(\approx 12500\) unknowns. This number of unknowns is well within the capabilities of modern hardware and optimization techniques.

\subsubsection{Stochastic gradient descent}
\label{sec:sgd}
Stochastic gradient descent (SGD) is often used instead of gradient decent when using large data sets. SGD aims to minimize objective functions of the form
\begin{equation*}
  J[\param] = \sum_{i=1}^{\ndata} g_{\param}(\x{}{i},\y{i}),
\end{equation*}
which includes the least-squares objective \eqref{eq:ls}. The objective function is updated using only one data point (batches can also be used),  which is chosen randomly at each iteration 
\begin{equation}
  \param \leftarrow \param - \eta \nabla g_{\param}(\x{}{i},\y{i}),
\end{equation}
Many variations on stochastic gradient descent have been developed. We refer the reader to \citep{Bottou2016} for more information. These variations often include adaptive strategies for choosing the learning rate $\eta$. Such methods have previous been applied in the context of tensors to minimize computation costs when dealing with large scale data \citep{Huang2015}.

We will use the adaptive strategy from ADAM \citep{Diederick2014}, as implemented in the $C^3$ library, to demonstrated the effectiveness of SGD in Section \ref{sec:experiments}.

\section{Gradient computation and analysis}
\label{sec:analysis}
The evaluation of the gradients of the FT with respect to its parameters is essential for the making gradient-based optimization algorithms tractable. Almost all existing literature for tensor approximation uses alternating minimization strategies because the subproblems are convex and can be solved exactly and efficiently with standard linear algebra algorithms, though they have limited (if any) guarantees for obtaining a good minimizer of the full problem. One of the major contributions of this paper is the derivation and analysis of the computational complexity of computing gradients of the supervised learning objective \eqref{eq:fspace_learn}.

The gradient of the least squares objective function with respect to a parameter is
\begin{equation} \label{eq:ls_grad}
  \frac{\partial J}{\partial \parami{\multi}} = \frac{2}{\ndata}\sum_{i=1}^{\ndata}\left( \y{i} - f\left(\x{}{i};\param\right)\right)  \frac{\partial f\left(\x{}{i};\param\right)}{\partial \parami{\multi}},
\end{equation}
where \(\multi = (k,i,j,\ell)\) denotes a multi-index for a unique parameter of the \(k\)th core, see Section \ref{sec:params}. Efficient computation of the partial derivative, \(\frac{\partial f(\x{};\param)}{\partial \parami{\multi}}\) is essential for making gradient based optimization algorithms tractable.
Letting
\begin{align}
  \mvf{F}_{<k}(\x{<k}) &= \mvf{F}_1(\x{1})\cdots\mvf{F}_{k-1}(\x{k-1}), \quad \textrm{ and } \label{eq:leftcores}\\
  \mvf{F}_{>k}(\x{>k}) &= \mvf{F}_{k+1}(\x{k+1})\cdots\mvf{F}_{d}(\x{d}). \label{eq:rightcores}
\end{align}
denote the products of the cores before and after the \(k\)th, respectively, and combining Equations \eqref{eq:ft}, \eqref{eq:leftcores}, and \eqref{eq:rightcores} we obtain the following expression
\begin{equation}\label{eq:ft_partial_deriv_param}
  \partial_{\multi}f(x) \equiv \frac{\partial f(\x{};\param)}{\partial \parami{\multi}} =\mvf{F}_{<k}(\x{<k}) \partial_{\multi} \mvf{F}_{k}(x_k) \mvf{F}_{>k}(\x{>k}),
\end{equation}
Thus efficiently computing gradients of the objective function and the FT requires efficient algorithms for evaluation of the left and right product cores \(\mvf{F}_{<k}(\x{<k})\) and  \(\mvf{F}_{>k}(\x{>k})\) and the partial derivative \(\partial_{\multi}\mvf{F}_{k}(x_k)\).

In the following  sections we describe and analyze the gradient of the FT with respect to its parameters. We first describe the partial derivatives of an FT with respect to the parameters of a single core. Then we present an efficient algorithm for computing left and right product cores. Finally, we discuss the computation of gradients of the full FT and the objective function.

\subsection{Derivatives of an FT core}
\label{sec:coregrad}
In this section we discuss how to obtain the partial derivatives with respect to parameters of a specific core. Without loss of generality , we will make the following assumption to ease the notational burden
\begin{assumption}\label{ass:numcoreparams}
  Each univariate function in each FT-core \eqref{eq:ftcore} is independently parameterized with $\pdimb$ parameters so that the FT core $\mvf{F}_k$ is parameterized with $r_{k-1} r_{k} \pdimb$ parameters.
\end{assumption}
Under this assumption the FT cores have the following structure.
\begin{proposition}\label{prop:sparsity}
  Let $\mat{G}_{kijl} \in \reals^{r_{k-1} \times r}$ denote the partial derivative $\frac{\partial \mvf{F}_k(z)}{\partial \paramuni{k}{i}{j}{\ell}}$ for some $z \in \xspace{k}$.  Under Assumption \ref{ass:numcoreparams}, $\mat{G}_{kijl}$ is a \textit{sparse} matrix with the following elements
  \begin{equation}
    \mat{G}[\alpha,\beta] = \left\{
    \begin{array}{cc}
      \frac{ \partial \unif{k}{i}{j}(x_k)}{\partial \paramuni{k}{i}{j}{\ell}} & \textrm{ if } \alpha=i,\beta=j \\
      0 & \textrm{ otherwise}
    \end{array} 
    \right.
  \end{equation}
  for $\alpha = 1, \ldots, r_{k-1}$, and $\beta = 1, \ldots r_{k}$.
\end{proposition}
Now if we let \(G(\pdimb)\) denote the number of operations required to compute the gradient of a univariate function with respect to \emph{all} of its parameters then the cost of computing \(\partial_{\multi} \mvf{F}_k(x_k)\) is \(   \mathcal{O}\left(G(\pdimb)r_{k-1}r_{k}\right)\) operations.

\subsection{Evaluating left and right product cores}
\label{sec:productcores}
Computation of the partial derivative of the FT \eqref{eq:ft_partial_deriv_param} requires the evaluation of the left and right product cores \(\mvf{F}_{<k}(\x{<k})\) and  \(\mvf{F}_{>k}(\x{>k})\).
A single forward and backward sweep the cores can be used to obtain these values using the following recursion identities
\begin{equation*}
  \mvf{F}_{<k+1}(x_{<k+1}) = \mvf{F}_{<k}(x_{<k})\mvf{F}_k(x_k) \quad \textrm{ and }  \quad
  \mvf{F}_{>k-1}(x_{>k-1}) = \mvf{F}_k(x_k)\mvf{F}_{>k}(x_{<k}).
\end{equation*}
Thus we can obtain the gradient of the FT with respect to parameters of the \(k\)th core using the Algorithms \ref{alg:leftsweep} and \ref{alg:rightsweep}.
\begin{algorithm}
  \caption{\texttt{coregrad-left}: Generate intermediate result for gradient computation}
  \label{alg:leftsweep}
  \begin{algorithmic}[1]
    \REQUIRE Number of rows, $r_{k-1}$;  Number of columns $r_{k}$; Number of parameters per univariate function $\pdimb$; FT core $\ftcore{k}$; Left multiplier $\bvec{a}$
    \FOR{$i = 1, \ldots r_{k-1}$}
        \FOR{$j = 1, \ldots r_{k}$}
            \FOR{$\ell = 1, \ldots, \pdimb$} 
            \STATE\label{alg:line:grad1}$\partial f_{kijl} =\bvec{a}[i]\frac{\partial \uni{f}{k}{i}{j}(x_k)}{\partial \paramuni{k}{i}{j}{\ell}}$
            \ENDFOR
        \ENDFOR
    \ENDFOR
  \end{algorithmic}
\end{algorithm}

\begin{algorithm}
  \caption{\texttt{coregrad-right}: Update intermediate gradient result with multiplier from the right}
  \label{alg:rightsweep}
  \begin{algorithmic}[1]
    \REQUIRE Number of rows, $r_{k-1}$;  Number of columns $r_{k}$; Number of parameters per univariate function $\pdimb$; Intermediate result $\partial f_{kijl}$; Right multiplier $\bvec{c}$
    \FOR{$i = 1, \ldots r_{k-1}$}
        \FOR{$j = 1, \ldots r_{k}$}
            \FOR{$\ell = 1, \ldots, \pdimb$} 
            \STATE $ \partial f_{kij\ell} =  \partial f_{kij\ell}\bvec{c}[j]$ \label{alg:line:finalgrad} 
            \ENDFOR
        \ENDFOR
    \ENDFOR
  \end{algorithmic}
\end{algorithm}
These two algorithms consist of a triple nested loop. The innermost loop of \texttt{coregrad-left} requires computing the gradient of a univariate function with respect to its parameters, and multipling each element of the gradient by the left multiplier. Thus if $r_k<r$ then $\mathcal{O}(G(p)r^2)$ operations are needed for the gradient and $\mathcal{O}(pr^2)$ products between floating point operations are needed. Finally, the partial derivative with respect to each parameter of the core is stored with a complexity of $\mathcal{O}(pr^2)$ floating point numbers. Each of these partial derivatives is updated in \texttt{coregrad-right} for a computational cost of $\mathcal{O}(pr^2)$ and no additional storage. Since each of these functions is called $\xdim$ times, the additional computational cost they incur is $\mathcal{O}(\xdim r^2(G(p) + p))$ and the additional storage complexity they incur is $\mathcal{O}( \xdim pr^2)$.

\subsection{Gradient of FT and objective functions}
\label{sec:fullgradient}

Th entire gradient of the FT can be obtained with a single forward and backward sweep as presented in Algorithm \ref{alg:gradeval}. In Algorithm \ref{alg:gradeval} we use \(\partial f_{k\multi_k}\) to denote the partial derivative of \(f\) with respect to the parameters of the \(k\)th core. Each step of the forward sweep requires: (i) creating an intermediate result for the gradient of the FT with respect to each parameter of the core in line \ref{alg:line:grad1} of Algorithm \ref{alg:leftsweep}, (ii) evaluating and storing the core of the FT in line \ref{alg:line:fteval} of Algorithm \ref{alg:gradeval}, and (iii) updating the product of the cores through the current step in line \ref{alg:line:leftup} of Algorithm \ref{alg:gradeval}). The backward sweep involves involves updating the the intermediate gradient result obtained from the forward sweep (line \ref{alg:line:finalgrad} in Algorithm \ref{alg:rightsweep}), and then updating the product of the cores in line \ref{alg:line:rightprod} of Algorithm \ref{alg:gradeval}.

\begin{algorithm}
  \caption{FT evaluation and gradient} 
  \label{alg:gradeval}
  \begin{algorithmic}[1]
    \REQUIRE Parameters $\param$; Evaluation location $x = (x_1,x_2,\ldots,x_{\xdim}) \in \xspace{}$
    \STATE \COMMENT{Generate intermediate results during forward sweep}
    \STATE $\bvec{a} = \mvf{F}_1(x_1)$ \label{alg:line:eval1}
    \STATE $\partial f_{1\multi_1} = \texttt{coregrad-left}(r_{0},r_{1},\pdimb,\ftcore{1},\bvec{a})$ 
    \FOR{$k = 2, \ldots, \xdim$}
        \STATE $\partial f_{k\multi_k} = \texttt{coregrad-left}(r_{k-1},r_{k},\pdimb,\ftcore{k},\bvec{a})$ 
   \STATE $\mat{F}_k = \mvf{F}_k(x_k)$ \label{alg:line:fteval} 
   \STATE $\bvec{a} \leftarrow \bvec{a}\mat{F}_k$ \label{alg:line:leftup} 
   \ENDFOR
   \STATE $f(x;\param) = \bvec{a}[1]$ \COMMENT{Evaluation of the FT}
    \STATE \COMMENT{Update intermediate results during backward sweep}
    \STATE $\bvec{c} = \mat{F}_{\xdim}$
    \FOR{$k=\xdim-1, \ldots, 2$}
   \STATE $\partial f_{k\multi_{k}} = \texttt{coregrad-right}(r_{k-1},r_{k},\pdimb,\partial f_{k\multi_{k}},\bvec{c})$ 
   \STATE $\bvec{c} \leftarrow \mat{F}_k\bvec{c}$  \label{alg:line:rightprod} 
    \ENDFOR
    \STATE $\partial f_{1\multi_{1}} = \texttt{coregrad-right}(r_{0},r_{1},\pdimb, f_{1\multi_{1}},\bvec{c})$ 
  \end{algorithmic}
\end{algorithm}

  \begin{proposition}\label{prop:grad}
    Let $f:\xspace{} \to \reals$ be a rank $\bvec{r} = [1\ r_1 \ldots r_{d-1} \ 1]$ FT with every univariate function in Equation \eqref{eq:ftcore}. Assume Assumption \ref{ass:numcoreparams} and that for $k = 1, \ldots, d-1$ we have $r_k < r$ for some $r \in \integers_+$. Let $E(\pdimb)$ denote the number of operations required to evaluate a univariate function. Let $G(\pdimb)$ denote the number of operations required to compute the gradient of a univariate function with respect to all of its parameters.  Then, evaluating $f(\x{};\param)$ and computing the gradient $\partial_{\multi} f(\x{};\param)$ with respect to its parameters requires
    \begin{equation*}
      \mathcal{O}\left(\xdim r^2\left(G(\pdimb) + E(p) + \pdimb\right)\right)
    \end{equation*}
    operations, and storing $\mathcal{O}(\xdim\pdimb r^2)$ floating point values.
  \end{proposition}
\begin{proof}
To demonstrate this result, we can calculate the number of evaluations and storage size required for each step of Algorithm~\ref{alg:gradeval}. First note that Lines~\ref{alg:line:eval1} and~\ref{alg:line:fteval} involve the evaluation of each FT core. Since each core has at most $r^2$ univariate functions these evaluations require $E(p)r^2$ operations and the ability to store $r^2$ floating point numbers. Since this evaluation has to happend for each of the \(d\) cores the computational complexity of this step is $dE(p)r^2$. The storage space can be reused and only requires storing two vectors of size \(r \times 1\) ($\bvec{a}$) and an \(r \times r\) matrix ($\mat{F}_k$) for a total storage complexity of $\mathcal{O}(r^2 + r) = \mathcal{O}(r^2)$ floating point numbers. Next we note that a product between a $1 \times r$ vector and an $r \times r$ matrix in lines~\ref{alg:line:leftup} and~\ref{alg:line:rightprod} needs to occur $2 \xdim$ times and therefore requires $\mathcal{O}(\xdim r^2)$ operations. Thus apart from the calls to $\texttt{coregrad-left}$ and $\texttt{coregrad-right}$ we have a computational complexity of $\mathcal{O}(\xdim E(p)r^2 )$ and a storage complexity of $\mathcal{O}(r^2)$. 

Combining these costs with the cost of the coregrad algorthims presented in Section \ref{sec:productcores} obtain the stated result.
\end{proof}

The values $G(p)$ and $E(p)$ are dependent upon the types of parameterizations of univariate functions. Consider two examples: one linear and one nonlinear. For both parameterizations we consider kernel basis functions; however, for the nonlinear parameterization we will consider the centers of each kernel as free parameters. The linear parameterization is given by Equation \eqref{eq:linparam} with \(\uni{\phi}{k\ell}{i}{j}(\x{k}) = \exp\left(-\frac{1}{\sigma^2}\left(\x{k} - c_{k\ell}\right)^2\right)\), where \(\sigma \in \reals_{+}\) and $c_{k\ell} \in \xspace{k}$ are the centers of the kernel such that each univariate function of the \(k\)th dimension is represented with as a sum of kernels with the same locations. If the evaluation of the exponential takes a constant number of operations with respect to $\x{k}$ then we have \(E(\pdimb) = \mathcal{O}(\pdimb)\) and \(G(\pdimb) = \mathcal{O}(\pdimb)\) because the the gradient of the univariate function with respect to its \(\ell\)th parameter is
\begin{equation}\label{eq:linunigrad}
\frac{ \partial \unif{k}{i}{j}(\x{k};\param)}{\partial \paramuni{k}{i}{j}{\ell}} = \exp\left(-\frac{1}{\sigma^2}\left(\x{k} - c_{k\ell}\right)^2\right).
\end{equation}
This gradient is independent of any other parameters. In practice this means that it can be precomputed at each location $\x{k}$ and recalled at runtime. In such a case the storage increases to $\mathcal{O}(\ndata \xdim  \pdimb r^2)$ numbers if each univariate function in each core has a different parameterization. If the univariate functions of each core share the same parameterization then the additional storage cost is  $\mathcal{O}(\ndata \pdimb \xdim)$. In either case the online cost becomes a simple lookup, i.e., \(G(p) = \mathcal{O}(1)\).

The nonlinear parameterization provided in Equation \eqref{eq:nonkernel} is different because we are free to optimize over the centers. In this case the gradient of each univariate function becomes
The gradient with respect to the second half of the parameters now depends on the parameter value
\begin{equation*}
  \frac{ \partial \unif{k}{i}{j}(\x{k};\param)}{\partial \paramuni{k}{i}{j}{\ell}}
   = 
\left \{
\begin{array}{cc}
\paramuni{k}{i}{j}{\ell} \exp\left(-\frac{1}{\sigma^2}\left(\x{k} - \paramuni{k}{i}{j}{(\pdimp{k}{i}{j}/2+\ell)}\right)^2 \right) & \textrm{ for } \ell = 1, \ldots, \pdimb/2 \\ 
\frac{2}{\sigma^2} \paramuni{k}{i}{j}{(\ell-\pdimb/2)}\left(\x{k} - \paramuni{k}{i}{j}{\ell}\right) \exp\left(-\frac{1}{\sigma^2}\left(\x{k} - \paramuni{k}{i}{j}{\ell}\right)^2 \right) & \textrm{ for } \ell = \pdimb/2+1, \ldots, \pdimb 
\end{array}
\right.
\end{equation*}
The parameter dependence of the gradient increases the complexity of optimization algorithms since precomputation of the gradients cannot be performed. In this case we have \(E(\pdimb) = \mathcal{O}(\pdimb)\) and \(G(\pdimb) = \mathcal{O}(\pdimb)\).

\subsection{Summary of computational complexity}
The complexity of evaluating the gradient of the least squares objective function 
\begin{equation*}
  \frac{\partial J}{\partial \parami{\multi}} = \frac{2}{\ndata}\sum_{i=1}^{\ndata}\left( \y{i} - f\left(\x{}{i};\param\right)\right)  \frac{\partial f\left(\x{}{i};\param\right)}{\partial \parami{\multi}}
\end{equation*}
is dominated by the cost of this derivative is evaluating the gradient of \(f\) at \(\ndata\) points. Consequently, the total complexity is \(\ndata\) times that provided by Proposition \ref{prop:grad}, for a total cost of \(\mathcal{O}(\ndata \xdim r^2(G(\pdimb)+E(\pdimb)))\) operations. The storage cost need not increase because one can evaluate the sum by sequentially iterating through each sample.

Now that we have described the computational complexity of the gradient computation, we summarize the computational complexity of the proposed optimization algorithms. Table \ref{tab:aao_complexity} shows the computational complexity of the all-at-once optimization scheme when using a low-memory BFGS solver. Three parameterizations are considered a linear parameterization with identical parameterizations of each univariate function in a particular core, a more general parameterization with varying linear parmaeterizations within each core and, and the most general case of nonlinear parameterization of each function. We see that using linear parameterizations allows us to precompute certain gradients before running the optimizer. This precomputation reduces the online cost of optimization. The computational complexity of the full solution is dominated by the number of evaluations and gradients of the objective function, and we denote this number as \(N_{\textrm{opt,AAO}}\). 
\begin{table}
\centering
\caption{Computational complexity of all-at-once optimization.}
\label{tab:aao_complexity}
\begin{tabular}[h]{|lccc|}
\hline
Parameterization type  &  Offline cost/storage &  Eval. and grad. (\(S\)) & Full solution \\
\hline
\hline
Shared/linear  & \(\mathcal{O}(\ndata \xdim G(\pdimb))\) / \(\mathcal{O}(\ndata \xdim \pdimb)\) & \(\mathcal{O}(\ndata \xdim r^2 \left( E(\pdimb) + \pdimb\right))\) &  \(\mathcal{O}(N_{\textrm{opt,AAO}}S)\)\\
General/linear & \(\mathcal{O}(\ndata \xdim r^2 G(\pdimb) )\) / \(\mathcal{O}(\ndata \xdim r^2 \pdimb)\) & \(\mathcal{O}(\ndata \xdim r^2 \left( E(\pdimb) + \pdimb\right) )\)  &  \(\mathcal{O}(N_{\textrm{opt,AAO}} S)\) \\  
nonlinear     & N/A &  \(\mathcal{O}\left(\ndata \xdim r^2\left( G(\pdimb) + E(\pdimb) + \pdimb\right)\right)\) &  \(\mathcal{O}(N_{\textrm{opt,AAO}} S) \) \\
\hline
\end{tabular}
\end{table}

The computational cost of stochastic gradient descent is given in Table \ref{tab:sgd_complexity}. In this case the cost per epoch (once through all of the training points) is the same as a single gradient evaluation in the all-at-once approach. The total cost of such a scheme is dominated by how many samples are used during the optimization procedures. In the associated table, we represent this number as the number of times one requires iterating through all of the samples \(N_{\textrm{epoch}} \ndata\). A fully online algorithm would not have any associated offline cost and its complexity would be equivalent to the nonlinear parameterization case. 
\begin{table}
\centering
\caption{Computational complexity of stochastic gradient descent with ADAM.}
\label{tab:sgd_complexity}
\begin{tabular}[h]{|lccc|}
\hline
Parameterization type  &  Offline cost/storage &  Eval. and grad. per sample (\(S\)) & Full solution\\
\hline
\hline
Shared/linear  & \(\mathcal{O}(\ndata \xdim G(\pdimb))\) / \(\mathcal{O}(\ndata \xdim \pdimb)\) & \(\mathcal{O}( \xdim r^2 \left( E(\pdimb) + \pdimb\right))\) & \(N_{\textrm{epoch}} \ndata S\)\\
General/linear & \(\mathcal{O}(\ndata \xdim r^2 G(\pdimb) )\) / \(\mathcal{O}(\ndata \xdim r^2 \pdimb)\) & \(\mathcal{O}( \xdim r^2 \left( E(\pdimb) + \pdimb\right) )\)  &  \(N_{\textrm{epoch}} \ndata S\)\\  
nonlinear    & N/A &  \(\mathcal{O}\left(\xdim r^2\left( G(\pdimb) + E(\pdimb) + \pdimb\right)\right)\) & \(N_{\textrm{epoch}} \ndata S\) \\
\hline
\end{tabular}
\end{table}

Finally, the alternating optimization scheme is of a slightly different nature. In the case of linear parameterization, each sub-optimization problem is quadratic and can be posed as solving a linear system by setting the right hand side of Equation \eqref{eq:ls_grad} to zero. This linear system has \(\ndata\) equations and \(\pdimb r^2\) unknowns and its solution, using a direct method, has an asymptotic complexity of \(\mathcal{O}\left(\ndata \pdimb^2 r^4\right) \). We also need to introduce a new constant called \(N_{\textrm{sweep}}\) that represents how many sweeps through all of the dimensions are required. We see that in the most general case, this algorithm is \(N_{sweep}\) times more expensive than all-at-once. However, this number is a bit desceptive since each sub problem has only \(\pdimb r^2\) unknowns and therefore we can assume that \(N_{\textrm{opt,ALS}} < N_{\textrm{opt,AAO}}\). In Table \ref{tab:als_complexity} we summarize the costs of this algorithm. 
\begin{table}
\centering
\caption{Computational complexity of alternating (non)-linear least squares.}
\label{tab:als_complexity}
\begin{tabular}[h]{|lccc|}
\hline
Param type  &  Offline cost/storage &  Solution of sub-problem (\(S\)) & Full solution\\
\hline
\hline
Shared/linear  & \(\mathcal{O}(\ndata \xdim G(\pdimb))\) / \(\mathcal{O}\left(\ndata \xdim \pdimb\right)\) & \(\mathcal{O}\left( \ndata \pdimb^2 r^4 \right)\) & \(\mathcal{O}\left( N_{\textrm{sweeps}}\xdim S \right)\)\\
General/linear & \(\mathcal{O}\left(\ndata \xdim r^2 G(\pdimb) \right)\) / \(\mathcal{O}\left(\ndata \xdim r^2 \pdimb\right)\) & \(\mathcal{O}\left( \ndata \pdimb^2 r^4 \right)\) & \(\mathcal{O}\left( N_{\textrm{sweeps}} \xdim S \right)\)\\
nonlinear   & N/A &  \(\mathcal{O}\left(N_{\textrm{opt,ALS}} \ndata r^2\left( G(\pdimb) + E(\pdimb) + \pdimb\right)\right)\) & \(\mathcal{O}\left( N_{\textrm{sweeps}} d S \right)\) \\
\hline
\end{tabular}
\end{table}

The summaries above were provided for the pure least squares regression problem. If we consider the regularization term of Equation \eqref{eq:tt_reg} then an additional step must be performed. The gradient of the functional \(\Omega[f]\) requires computing the gradient of each of the summands, which can be written as
\begin{align*}
& \frac{\partial \int_{\xspace{k}} \left[\unif{k}{i}{j}\left( x_k;\paramuni{k}{r_{k-1}}{r_{k}}{1}, \ldots, \paramuni{k}{r_{k-1}}{r_{k}}{\pdimb}\right) \right]^2  dx_k}{\partial \paramuni{k}{r_{k-1}}{r_k}{\ell}} = \\
& \quad \quad \quad  2 \int_{\xspace{k}} \frac{\partial \unif{k}{i}{j}\left( x_k;\paramuni{k}{r_{k-1}}{r_{k}}{1}, \ldots, \paramuni{k}{r_{k-1}}{r_{k}}{\pdimb}\right)}{\partial \paramuni{k}{r_{k-1}}{r_k}{\ell}}\unif{k}{i}{j}(x_k;\paramuni{k}{r_{k-1}}{r_{k}}{1}, \ldots, \paramuni{k}{r_{k-1}}{r_{k}}{\pdimb})  dx_k,
\end{align*}
for \(\ell = 1, \ldots, \pdimb\). In other words, to compute the gradient of every element in the sum in Equation \eqref{eq:tt_reg}, one computes the inner product between the original function and the function representing the partial derivative of this function. In the case of linear parameterizations described above, the partial derivative is simply a scalar and only the integral of the univariate function needs to be computed. Alternatively the full inner product must be evaluated; however we note that the evaluation of this integral is often analytical or available in closed form based upon the type of function. For example the cost is \(\mathcal{O}(\pdimb)\) for orthonormal basis functions and \(\mathcal{O}(\pdimb^2\)) more generally, to obtain the gradient with respect to every parameter of a univariate function.

\section{Experiments}
\label{sec:experiments}

In this section, we provide experimental validation of our approach. Our validation is performed on a synthetic function, approximation benchmark problems, and several real-world data sets. The synthetic examples are used to show the convergence of our approximations with increasing number of samples using the various optimization approaches. We also include comparisons between both the nonlinear parameterized FT and the linearly parameterized FT-c representations. Furthermore, we show the effectiveness of our rank adaptation approach.  The real-world data sets indicate the viability of FT-based regression for a wide variety of application areas and indicates pervasiveness of low-rank structure. The Compressed Continuous Computation ($C^3$) library \citep{c3} was used for all of experiments and is publicly available.

\subsection{Comparison of optimization strategies}
We first compare the convergence, with the number of samples, of three optimization algorithms. We use three synthetic test cases with known rank and parameterization. The first two functions are from a commonly used benchmark database \citep{simulationlib}. The third function is a FT-rank 2 function that is commonly used to demonstrate the performance of low-rank algorithms.

The first function is six dimensional and called the OTL circuit function. It is given by
\begin{equation}
f(R_{b1}, R_{b2}, R_{f}, R_{c1}, R_{c2}, \beta) = \frac{(V_{b1} + 0.74)\beta (R_{c2} + 9)}{\beta (R_{c2} + 9) + R_f} + \frac{11.35R_f}{\beta(R_{c2} + 9) + R_f} + \frac{0.74  R_f  \beta  (R_{c2} + 9)}{(\beta  (R_{c2} + 9) + R_f)R_{c1}}, 
\label{eq:otl}
\end{equation}
with \(V_{b1} = \frac{12R_{b2}}{R_{b1} + R_{b2}}\) and variable bounds \(R_{b1} \in [50, 150]\), \(R_{b2} \in [25, 70]\), \(R_{f} \in [0.5, 3]\), \(R_{c1} \in [1.2, 2.5]\), \(R_{c2} \in [0.25, 1.2],\) and \(\beta \in [50, 300]\). This function has a decaying spectrum due to the complicated sum of variables in the denominators, and it provides and an important test problem for our proposed rank adaptation scheme. Before exploring rank-adaptation first we explore the effectiveness of the three optimization algorithms in the context of fixed rank \(r\) and number of univariate parameters \(\pdimb\). Specifically we compute the relative squared error, using 10000 validation samples, for increasing number of training samples. We use a stopping criterion of \(10^{-13}\) for the difference between objective values for the gradient based techniques and for the difference between functions of successive sweeps for ALS. Though we have found that the results for these functions are not too sensitive to this tolerance. We solve 100 realizations of the optimization problem for each sample size and report the median, 25th, and 75th quantiles. The univariate function are parameterized using Legendre polynomials because we are using a domain with uniform measure.

Figure \ref{fig:otl_comparison} demonstrates that stochastic gradient descent and the all-at-once approach tend to outperform alternating least squares. This performance benefit is greater in regions of small sample sizes and large number of unknowns (large \(r\) and \(p\)). In the case of the largest number of unknowns in Figure \ref{fig:otl_4_8} our gradient based methods obtain an error that is several orders of magnitude lower error than the error obtained by ALS.
\begin{figure}
  \centering
  \begin{subfigure}[b]{0.3\textwidth}
    \includegraphics[width=\linewidth]{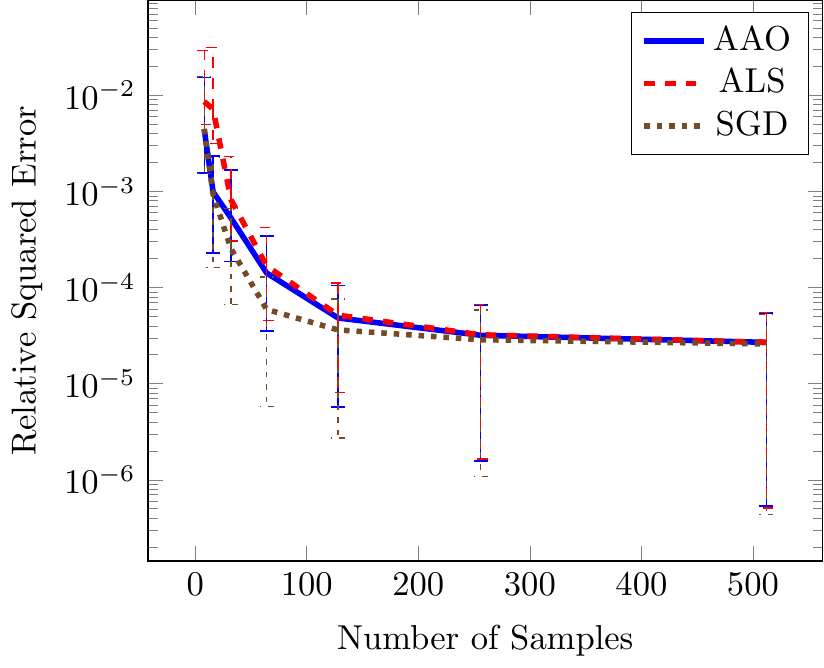}
    \caption{$r=2, p=3$}
    \label{fig:otl_2_2}
  \end{subfigure}
  ~ 
  \begin{subfigure}[b]{0.3\textwidth}
    \includegraphics[width=\linewidth]{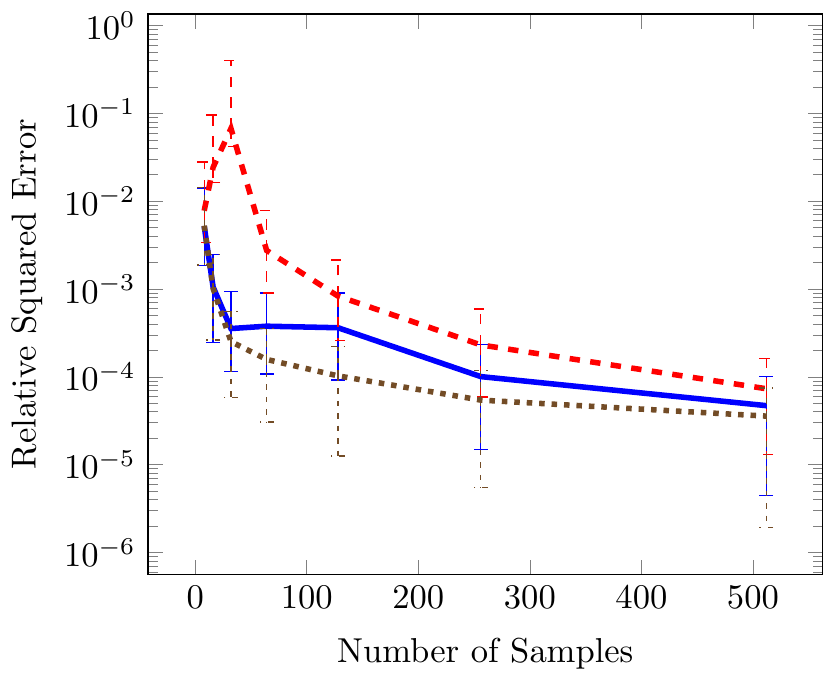}
    \caption{$r=4, p=3$}
    \label{fig:otl_4_2}
  \end{subfigure}
  ~ 
  \begin{subfigure}[b]{0.3\textwidth}
    \includegraphics[width=\linewidth]{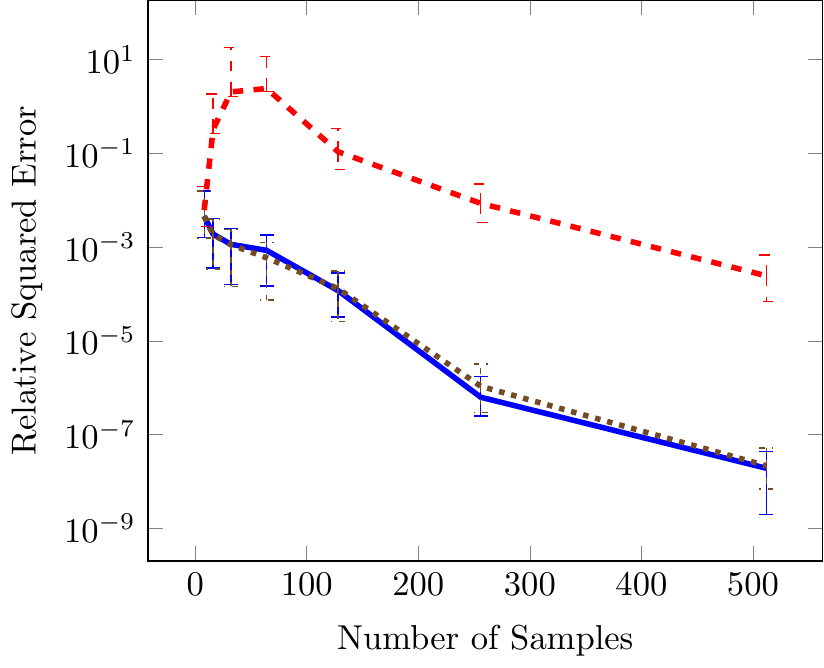}
    \caption{$r=4, p=9$}
    \label{fig:otl_4_8}
  \end{subfigure}
  \caption{Median, 25th, and 75th quantiles of relative error over 100 realizations of training samples for the OTL Circuit~\eqref{eq:otl}.}
  \label{fig:otl_comparison}
\end{figure}

Next we consider two cases in which the rank is known. The wing weight function is ten dimensional and given by 
\begin{align}
f(S_w,W_{fw},A,&\Lambda,q,\lambda,t_c,N_z,W_{dg},W_{p}) =  \nonumber\\
& 0.036S_{w}^{0.758}W_{fw}^{0.0035} \left( \frac{A}{\cos^{2}(\Lambda)}\right)^{0.6}q^{0.006}\lambda^{0.04}\left(\frac{100t_c}{\cos(\Lambda)}\right)^{-0.3}(N_zW_{dg})^{0.49} + S_wW_p, \label{eq:ww}
\end{align}
with variable bounds \(S_w \in [150, 200]\), \(W_{fw} \in [220,300]\), \(A \in [6,10]\), \(\Lambda \in [-10, 10]\), \(q \in [16, 45]\), \(\lambda \in [0.5, 1]\), \(t_c \in [0.08, 0.18]\), \(N_z \in [2.5, 6]\), \(W_{dg} \in [1700, 2500]\), \(W_{p} \in [0.025, 0.08]\). Using the variable ordering above the rank of this function is \(r=2\). The results in Figure \ref{fig:ww_comparison} indicate the same qualitative performance of the three optimizaiton methods. However, the difference in this case is that the SGD is not significantly better than all-at-once (AAO) approach for low sample sizes. In the third panel we see that SGD levels off around a relative squared error of \(10^{-9}\). For such small errors we have found it difficult to converge the SGD to smaller errors because of the tuning parameters involved in ADAM. In particular, the final error tolerance becomes sensitive to the choice of initial learning rate.
\begin{figure}
  \centering
  \begin{subfigure}[b]{0.3\textwidth}
    \includegraphics[width=\linewidth]{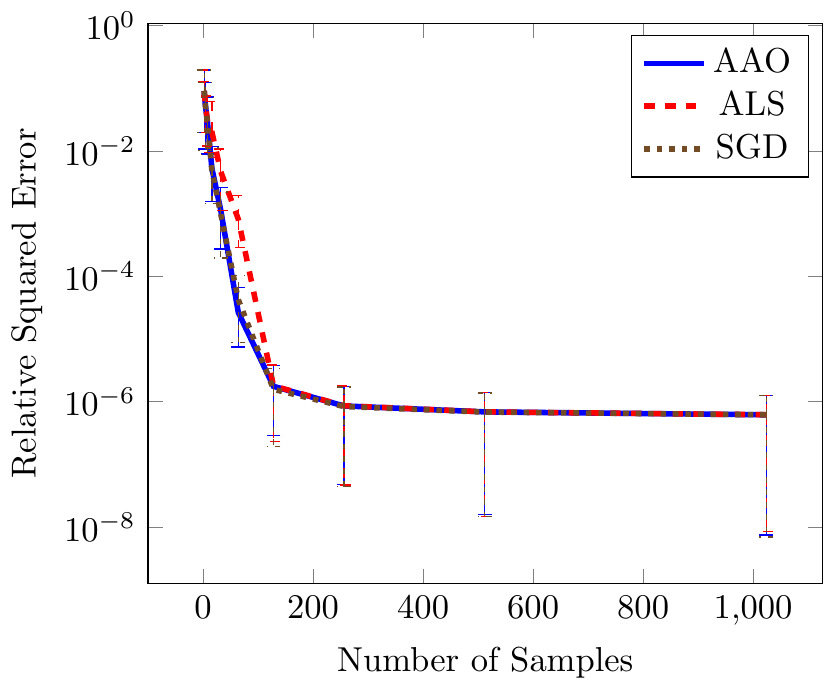}
    \caption{$p=3$}
    \label{fig:ww_2_2}
  \end{subfigure}
  ~ 
  \begin{subfigure}[b]{0.3\textwidth}
    \includegraphics[width=\linewidth]{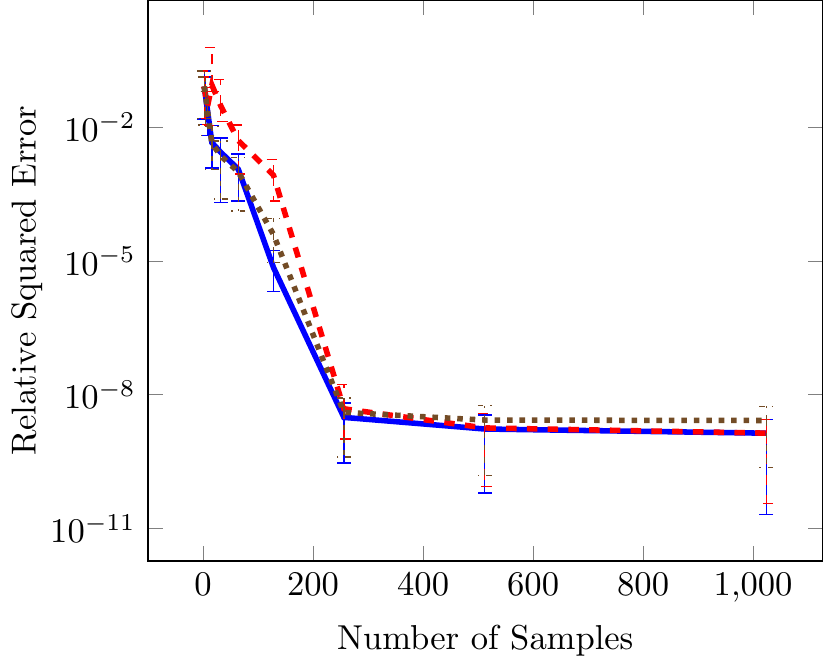}
    \caption{$p=5$}
    \label{fig:ww_4_2}
  \end{subfigure}
  ~ 
  \begin{subfigure}[b]{0.3\textwidth}
    \includegraphics[width=\linewidth]{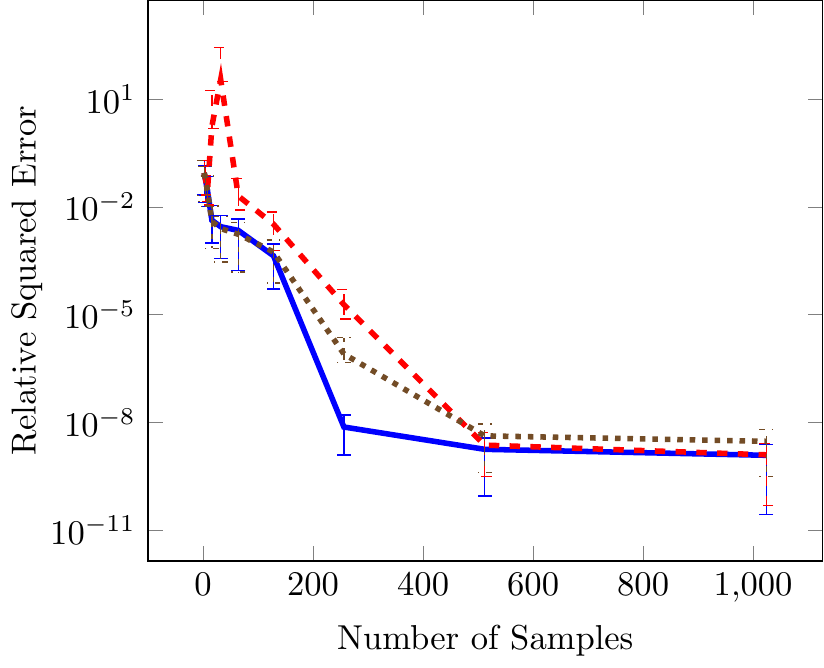}
    \caption{$p=6$}
    \label{fig:ww_4_5}
  \end{subfigure}
  \caption{Median, 25th, and 75th quantiles of relative error over 100 realizations of training samples for the Wing Weight function~\eqref{eq:ww}.}
  \label{fig:ww_comparison}
\end{figure}

The final test function is six dimensional and commonly used for testing tensor approximation because it has a TT-rank of 2 \citep{Oseledets2010}, this function is a Sine of sums
\begin{equation}
f(x) = \sin\left(\sum_{i=1}^6 x_i\right), \quad x_i \in [-1, 1], \quad i = 1, \ldots, 6.
\label{eq:sinsum}
\end{equation}
In Figure \ref{fig:sin_sum} we also see that the gradient based approaches are more effective in the case of small number of data sets, but that all achieve essentially the same minimums as the number of samples increase.

\begin{figure}
  \centering
  \begin{subfigure}[b]{0.3\textwidth}
    \includegraphics[width=\linewidth]{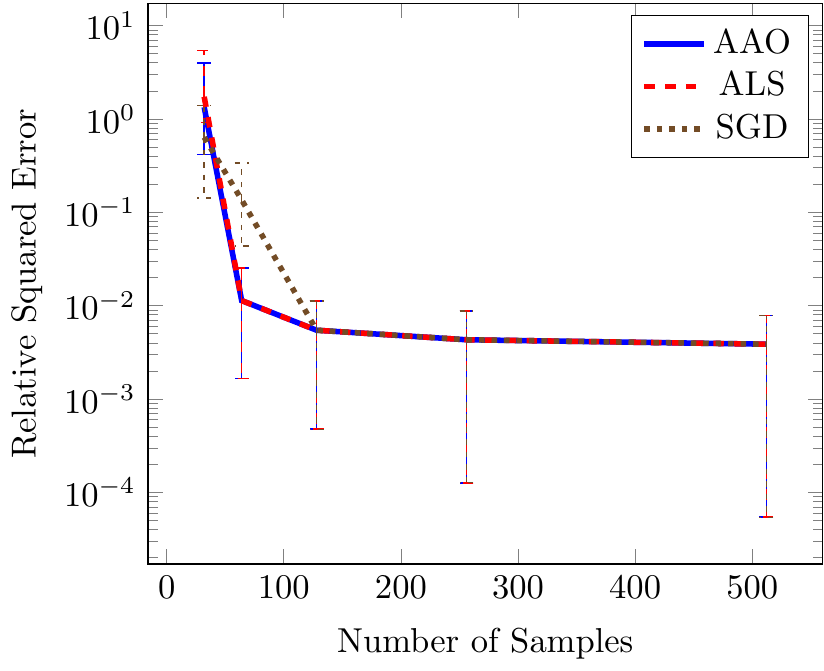}
    \caption{$p=3$}
    \label{fig:sinsum_2_2}
  \end{subfigure}
  ~ 
  \begin{subfigure}[b]{0.3\textwidth}
    \includegraphics[width=\linewidth]{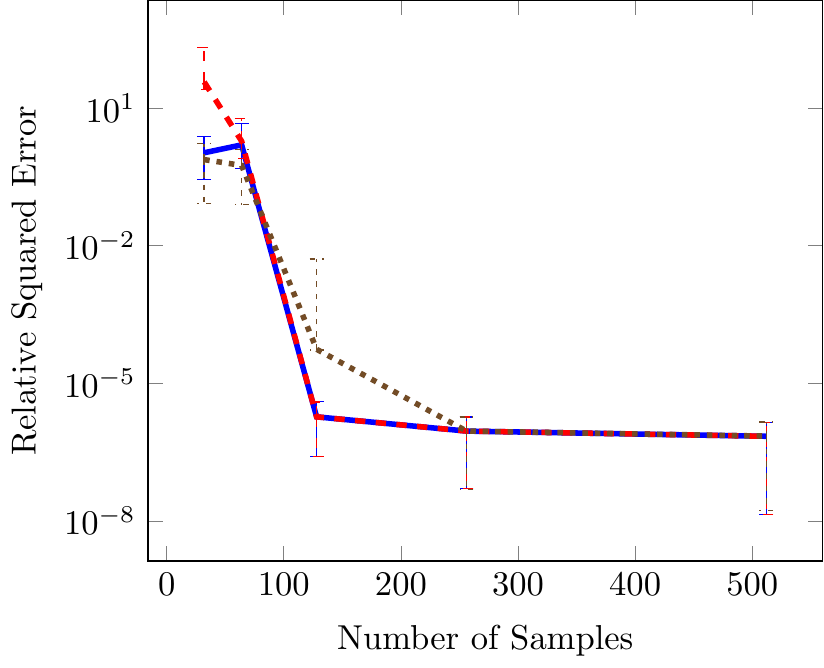}
    \caption{$p=5$}
    \label{fig:sinsum_2_4}
  \end{subfigure}
  ~ 
  \begin{subfigure}[b]{0.3\textwidth}
    \includegraphics[width=\linewidth]{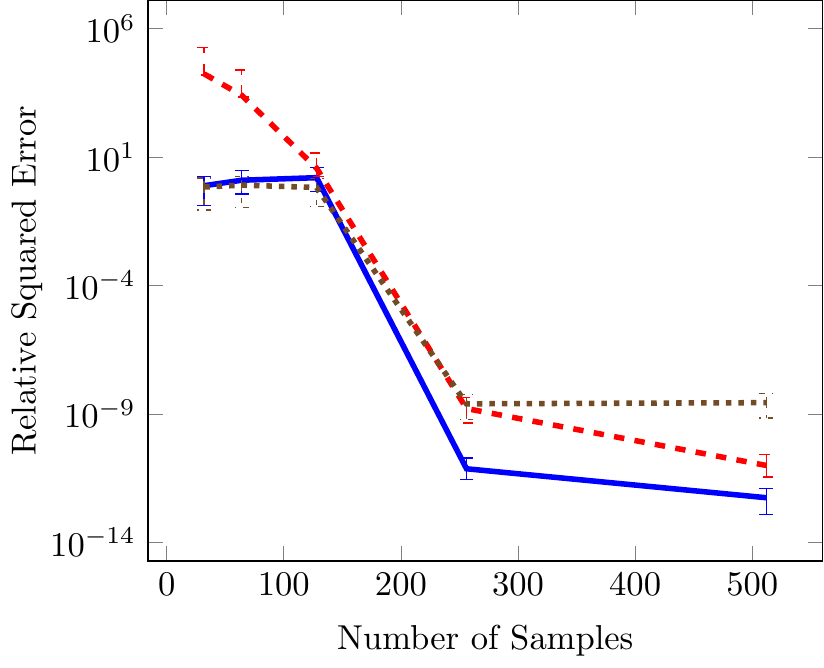}
    \caption{$p=9$}
    \label{fig:sinsum_2_8}
  \end{subfigure}
  \caption{Median, 25th, and 75th quantiles of relative error over 100 realizations of training samples for the Sine of sums~\eqref{eq:sinsum}.}
  \label{fig:sin_sum}
\end{figure}

\subsection{Linear vs nonlinear approximation}
Next we compare the FT and FT-c representations with different basis functions.  Specifically, using the OTL function \eqref{eq:otl}, we compare kernels at fixed locations and with kernels at optimized locations. For the linear approximation we use 8 kernels with fixed centers, and for the nonlinear approximation we use 4 kernels with moving centers (for a total of \(\pdimb=8\) for both approximation types). The results of this study are shown in Figure \ref{fig:otl_poly_vs_kernel}. Using AAO optimization we see that, for this problem, the nonlinear parameterization of kernels with moving centers provides a more effective representation. Specifically, we achieve an order of magnitude reduction in error when using the nonlinear moving-center representation.
\begin{figure}
  \centering
  \includegraphics[width=0.4\textwidth]{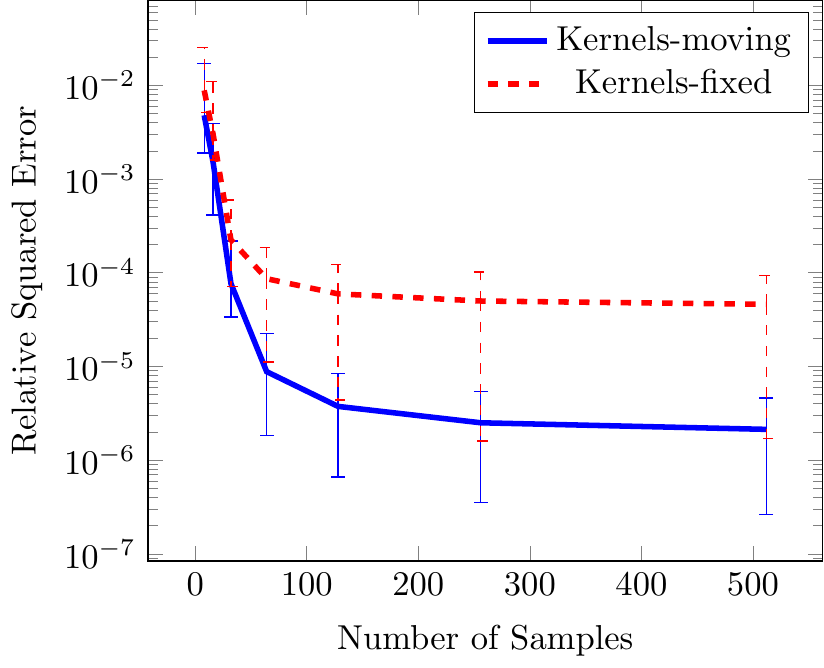}
  \caption{Comparison of parameterization with eight kernels with fixed and uniformly distributed centers and four kernels with with optimized centers with rank $r=4$. Median, 25th, and 75th quantiles are obtained over 100 realizations of training data. Experiments are performed on Equation \eqref{eq:otl}.}
  \label{fig:otl_poly_vs_kernel}
\end{figure}

\subsection{Rank adaptation}
Since the OTL circuit function \eqref{eq:otl} does not have finite rank, it is a good candidate for exploration of our rank adaptation scheme. With this goal, we compare rank adaptation with three approximations with varying fixed-rank, while fixing the polynomial order to eight. The results, shown in Figure \ref{fig:otl_rank_adapt}, indicate the desired behavior of the rank adaptation scheme, which follows the line of the best fixed-rank approximation. In particular for low sample sizes the rank-2 approximation is the best fixed-rank approximation, and the adaptive scheme shows approximately the same error. Once the sample size increases, the higher rank approximations converge more quickly and the adaptive scheme is able to follow the best approximation.
\begin{figure}
  \centering
  \includegraphics[width=0.4\textwidth]{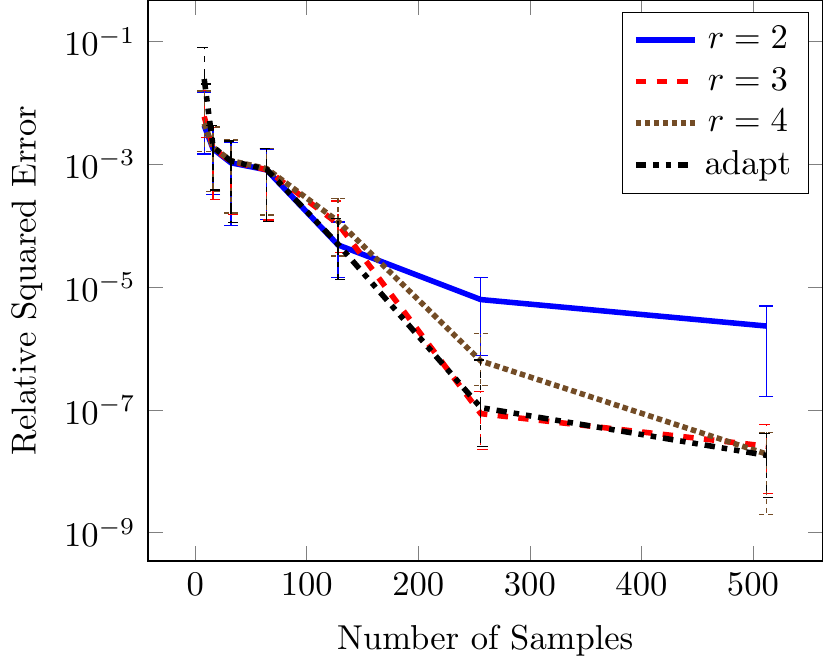}
  \caption{Rank adaptation performance, polynomial order is fixed to eight. Median, 25th, and 75th quantiles are obtained over 100 realizations of training data. The adaptive scheme follows the lowest-error-curve for each of the sample size. Experiments are performed on Equation \eqref{eq:otl}.}
  \label{fig:otl_rank_adapt}
\end{figure}

\subsection{Real-data performance}

In this section we compare regression in the FT format on real-world data sets and with other regression algorithms. In particular, we use 10 real-world data sets made available by \cite{Kandasamy2016} for which a set of 22 algorithms, 19 nonparametric and 3 parametric, were compared. The three parametric algorithms included a ridge regularized regression algorithm with linear basis functions, and two sparse regularized algorithms LASSO and LAR. These data sets are sourced from a variety of sources and preprocessed to normalize the inputs and outputs to zero mean and a standard deviation of one. To allow estimation of prediction error, the authors randomly split each data set in half to generate a training sample set and a separate validation set. The mean squared error 
\begin{equation*}
MSE = \frac{1}{N}\sum_{i=1}^{\ndata_{\textrm{validation}}} \left(\hat{f}(\x{}{i}) - \y{i}\right)^{2},
\end{equation*}
is calculated only over the validation set. 

In addition to comparing the FT to the other algorithms, we also compare with a robust and effective parameteric approximation scheme based on cross-validated LASSO. This scheme is able to use higher-order polynomials and we have found it to perform better than the results reported in the paper. Let \(\Phi\) be a Vandermonde-like matrix whose entries \(\Phi_{ij}=\phi_j(x_i)\) are the \(j\)-th basis function evaluated at the \(i\)-th point, then LASSO finds the basis coefficients that minimize
\begin{equation}\label{eq:lasso}
\lVert\Phi\param\rVert^2_2 + \lambda_\text{LASSO} \lVert\param\rVert_1
\end{equation}
We use least angle regression \cite{Tibshirani1996} to solve the LASSO problem and use 10-fold cross validation to choose the regularization parameter \(\lambda_\text{LASSO}\). We also use cross validation to simulatenously choose the degree of the total-degree polynomial basis. Only degree-one and degree-two polynomial spaces were considered because the size of the Vandermonde matrix in linear-parameteric representations grows exponentially with dimension. Note that these examples highlight the fact that we can use expressive basis functions in high high dimensions by exploiting low-rank structure. In other words, low-rank functional decompositions enable more expressive parameterized approximation forms.

For the FT, we use a kernel basis and perform 20-fold cross validation to choose the number of kernels \(\pdimb \in \{3, 6, 9\}\), the rank \(r \in \{1, 2\}\), the kernel width parameters \(w \in \{1,2,4,6,8\}\), and the regularization term \( \lambda \in \{10^{-3}, 10^{-7}\} \). The kernel width is chosen similarly to \citep{Kandasamy2016} where we have \(\sigma = w \ndata^{1/5} \hat{\sigma}\), and \(\hat{\sigma}\) denotes the standard deviation of the input data. Because we do not know the input domains for this data, therefore we position the kernels according to the empirical marginal distribution the data sets along each dimension. Specifically, we position the kernels at uniform quantiles of the data between the 10th and 90th quantiles. Finally, we use the AAO optimization setup, and report the mean squared errors on the testing data in Table \ref{tab:rd_res}. The FT is the best model for 2 data sets and in top five for seven data sets, it is the only parametric model that scored the best on at least one data set. The FT also outperforms the four other parametric models that use sparse regression or ridge regression. 
Because these data sets come from a wide variety of application areas, these results indicate that low-rank structure exists and is pervasive in a wide variety of regimes.

\begin{table}
\centering
\footnotesize
\caption{Mean squared error on validation data for a subset of regression algorithms and data sets from \cite{Kandasamy2016}. Only the algorithms which scored the best in at least one data set and the CV-based LASSO are shown. The FT is in the last column and highlighted in blue. Bolded numbers indicate the best performance for each data set and underlined numbers indicate a top five performance out of all 22 algorithms.}
\label{tab:rd_res}
\begin{tabular}[htb]{|c|c|c|c|c|c|c|c|}
\hline Dataset $(d,n)$ &SALSA &nSVR &RT &GBRT &MARS &\cellcolor{blue!25}FT &CVLASSO \\ 
\hline Housing (12,256) & \underline{\bf{0.26241}} & 0.38600 & 1.06015 & 0.42951 & 0.42379 & \underline{0.32798} & \underline{0.35218} \\ 
\hline Galaxy (20,2000) & \underline{\bf{0.00014}} & 0.15798 & 0.02293 & 0.01405 & \underline{0.00163} & \underline{0.00056} & 0.00249 \\ 
\hline Skillcraft (18,1700) & \underline{0.54695} & 0.66311 & 1.08047 & 0.57273 & \underline{0.54595} & \underline{\bf{0.54434}} & 0.56879 \\ 
\hline CCPP (59,2000) & 0.06782 & 0.09449 & 1.04527 & \underline{\bf{0.06181}} & 0.08189 & \underline{0.06631} & \underline{0.06466} \\ 
\hline Speech (21,520) & \underline{0.02246} & 0.06994 & 0.05430 & 0.03515 & \underline{\bf{0.01647}} & \underline{0.02684} & 0.03131 \\ 
\hline Music (90,1000) & \underline{0.62512} & \underline{\bf{0.59399}} & 1.45983 & 0.66652 & 0.88779 & 0.72094 & 0.64485 \\ 
\hline Telemonit (19,1000) & \underline{0.03473} & 0.05246 & \underline{\bf{0.01375}} & 0.04371 & \underline{0.02400} & 0.04040 & 0.06487 \\ 
\hline Propulsion (15,200) & \underline{0.00881} & 0.00910 & 0.02341 & \underline{0.00061} & 0.01290 & \underline{\bf{0.00009}} & 0.02660 \\ 
\hline Airfoil (40,750) & 0.51756 & 0.55118 & \underline{0.45249} & \underline{\bf{0.34461}} & 0.54552 & \underline{0.46920} & \underline{0.46444} \\ 
\hline Forestfires (10,211) & \underline{0.35301} & 0.43142 & 0.41531 & \underline{\bf{0.26162}} & \underline{0.33891} & 0.39465 & 0.44175 \\ 
\hline
\end{tabular}
\end{table}

\subsection{Application: modeling a propulsion plant on a naval vessel}

In this section we consider a simulation of a gas turbine propulsion engine mounted on a naval Frigate as detailed in \cite{Coraddu2013}, and for which simulation data is made openly available through the UCI Machine Learning Repository \citep{Lichman2013}. The goal of the UQ problem is to predict the degradation of the gas turbine based on certain parameters of its operation. According to \cite{Coraddu2013}, the model for the propulsion system is made of three components: the engine, the transmission gear, and the propulsor. The model is described by a set of nonlinear differential equations.

This simulation has sixteen parameters as detailed in Table \ref{tab:params}. The output that we attempt to predict using these simulation inputs is the gas turbine degradation coefficient that describes the gas flow rate reduction factor over service hours.
\begin{table}
\caption{Regression inputs to propulsion simulator from \cite{Coraddu2013}. HP denotes high-pressure and GT denotes gas turbine.}
\label{tab:params}
\centering
\begin{tabular}{lcc}
\hline
Variable &  Units \\
\hline
\hline
Lever position & (\(\cdot\)) \\
Ship speed & knot \\
Gas turbine shaft torque & kN m \\
Gas turbine rate of revolutions & r/min \\
Gas generator rate of revolutions & r/min \\
Starboard propeller torque & kN \\
Port propeller torque & kN \\
HP turbine exit temperature & C \\
GT compressor inlet air temperature & C \\ 
GT compressor outlet air temperature & C \\
HP turbine exit pressure & bar \\
GT compressor inlet air pressure & bar \\
GT compressor outlet air pressure & bar \\
GT exhaust gas pressure & bar \\
Turbine injection Control & \% \\
Fuel flow & kg / s\\
\hline
\end{tabular}
\end{table}

We use the data provided through the UCI repository to compare our proposed low-rank regression methodology to commonly used sparse regression algorithms. This data consists of 11934 instances of parameters and outputs. We also noticed that as part of this data the GT compressor inlet temperature and the GT compresser inlet air pressure did not vary, but we still included these variables in the regression problem to check if our approaches are robust to such cases.

We use the CV-based LASSO scheme described above using a total order basis of up-to 4th order polynomials and we use a rank adaptive low-rank regression scheme using the AAO approach. We also cross validate for up to 4th order polynomials using 5-fold cross validation, and we limit the BFGS algorithm to 500 iterations. We perform regression for 20 realizations of training data. We consider training sample sizes of  \(29, 59, 119 \), and 238 samples, and we validate on the remaining data. 

Figure \ref{fig:propulsion_results} demonstrate that we are able to achieve several orders of magnitude reduction in MSE with the low-rank approach as compare to the sparse regression approach for this problem. These results suggest that high frequency interactions exist between the input parameters that are not captured using a total order polynomial expansion. Indeed a full tensor product basis is needed, and that the coefficients of this tensor product basis are low rank. Performing sparse regression with the full tensor product basis of 4th order polynomials would have required solving for \(5^{16} = \mathcal{O}(10^{11})\), or approximately 152 billion unknowns.

\begin{figure}
  \centering
  \includegraphics[width=0.4\textwidth]{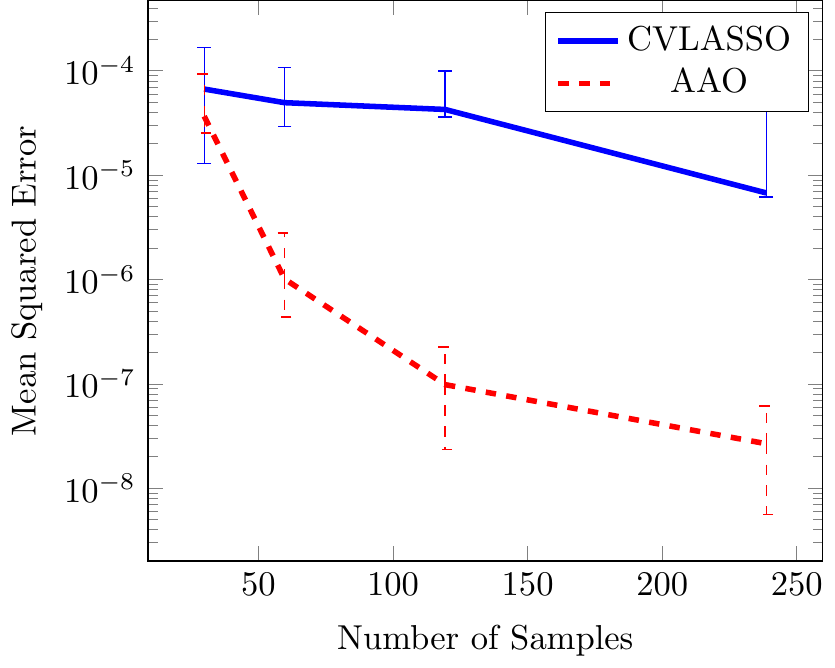}
  \caption{Comparison of rank adaptive low-rank regression with AAO optimization and CV LASSO on simulated data from the gas turbine of a naval vessel propulsion plant from \cite{Coraddu2013}.}
  \label{fig:propulsion_results}
\end{figure}

\section{Conclusions}

We have derived and analyzed the computational complexity of gradient based optimization for regression in a low-rank functional tensor format, the functional tensor-train (FT). Our analysis is valid for both the common FT-c variant, where a tensor-train of coefficients of a tensor product basis is used to represent the function, and for the more general case where the FT is represented using a set of (non)linearly parameterized univariate functions. Furthermore, we have proposed and demonstrated the effectiveness of both a rank-adaptation scheme to prevent overfitting and nonlinear parameterizations of univariate functions to increase expressivity.

Our results indicate that full gradient-based approaches have significant accuracy advantages over the standard practice of alternating least squares optimization. The accuracy of full gradient based schemes is especially improved in the context of low sample sizes with respect to the number of free parameters. These empirical findings hold true for both batch gradient optimization methods such as BFGS and stochastic gradient descent methods such as ADAM. Furthermore, we have shown that low-rank approximation itself is a promising model for general function approximation and machine learning. In particular, tests on 10 data sets from various application areas indicate that the FT is competetitive, and sometimes better than, 22 other commonly used algorithms.

There exists many directions for future research. The first direction is improving the performance for noisy and/or small data through more effective regularization techniques. Here, we have incorporated a basic group sparsity regularization term, but this area of research is actively being developed and can be expanded into the multilinear context we consider here. Another direction for research is reducing the need to cross validate over the number of parameters in each basis through more adaptive techniques that modify the number of parameters on the fly.

\section*{Acknowledgements}
This work was supported by the DARPA EQUIPS project and by the John von Neumann Postdoctoral Fellowship offered through the U.S. Department of Energy, Office of Science, Office of Advanced Scientific Computing Research, Applied Mathematics program. Sandia National Laboratories is a multimission laboratory managed and operated by National Technology and Engineering Solutions of Sandia, LLC., a wholly owned subsidiary of Honeywell International, Inc., for the U.S. Department of Energy's National Nuclear Security Administration under contract DE-NA-0003525.

\bibliographystyle{unsrtnat}

\bibliography{../bibliography/references}
\end{document}